\documentclass[a4paper,english,cleveref, autoref, thm-restate,authorcolumns]{lipics-v2019}

\usepackage{graphicx}
\usepackage{epsfig}
\usepackage{epstopdf}
\usepackage[usenames,dvipsnames]{xcolor}
\usepackage{braket}
\usepackage{amsthm, mathtools}
\usepackage{bm}
\usepackage{framed}
\usepackage{dsfont}
\usepackage{bbold}
\usepackage{eurosym}
\usepackage{microtype}
\usepackage{hyperref}

\allowdisplaybreaks

\DeclareMathOperator{\poly}{poly}
\DeclareMathOperator{\re}{Re}
\DeclareMathOperator{\im}{Im}

\newlength{\leftbarwidth}
\newlength{\leftbarsep}
\colorlet{leftbarcolor}{black}

\renewenvironment{leftbar}{%
    \MakeFramed {\advance \hsize -\width \FrameRestore }%
}{%
    \endMakeFramed
}

\newcommand{\be}{\begin{equation}}
\newcommand{\ee}{\end{equation}}
\newcommand{\ba}{\begin{aligned}}
\newcommand{\ea}{\end{aligned}}

\newcommand{\bc}{\begin{center}}
\newcommand{\ec}{\end{center}}
\newcommand{\beq}{\begin{equation}}
\newcommand{\eeq}{\end{equation}}
\newcommand{\beqq}{\begin{equation*}}
\newcommand{\eeqq}{\end{equation*}}
\newcommand{\beqa}{\begin{align}}
\newcommand{\eeqa}{\end{align}}
\newcommand{\barr}{\begin{array}}
\newcommand{\earr}{\end{array}}
\newcommand{\bi}{\begin{itemize}}
\newcommand{\ei}{\end{itemize}}

\newcommand{\Tr}{\ensuremath{\,\mathrm{Tr}}}
\newcommand{\Prb}{\ensuremath{\,\mathrm{Pr}}}

\newtheorem{lem}{Lemma}
\newtheorem{theo}{Theorem}

\newtheorem{coro}{Corollary}

\title{Building trust for continuous variable quantum states} 

\titlerunning{Building trust for continuous variable quantum states} 

\author{Ulysse Chabaud\footnote{Corresponding author}}{Laboratoire d'Informatique de Paris 6, CNRS, Sorbonne Universit\'e, 4 place Jussieu, 75005 Paris, France}{ulysse.chabaud@gmail.com}{https://orcid.org/0000-0003-0135-9819}{}

\author{Tom Douce}{School of Informatics, University of Edinburgh, 10 Crichton Street, Edinburgh, EH8 9AB, United Kingdom}{}{}{}

\author{Fr\'ed\'eric Grosshans}{Laboratoire d'Informatique de Paris 6, CNRS, Sorbonne Universit\'e, 4 place Jussieu, 75005 Paris, France and Laboratoire Aim\'e Cotton, CNRS, Universit\'e Paris-Sud, ENS Cachan, Universit\'e Paris-Saclay, 91405 Orsay Cedex, France}{}{https://orcid.org/0000-0001-8170-9668}{}

\author{Elham Kashefi}{Laboratoire d'Informatique de Paris 6, CNRS, Sorbonne Universit\'e, 4 place Jussieu, 75005 Paris, France and School of Informatics, University of Edinburgh, 10 Crichton Street, Edinburgh, EH8 9AB, United Kingdom}{}{}{}

\author{Damian Markham}{Laboratoire d'Informatique de Paris 6, CNRS, Sorbonne Universit\'e, 4 place Jussieu, 75005 Paris, France}{}{}{}

\authorrunning{U. Chabaud, T. Douce, F. Grosshans, E. Kashefi and D. Markham} 

\Copyright{Ulysse Chabaud, Tom Douce, Fr\'ed\'eric Grosshans, Elham Kashefi and Damian Markham} 



\category{} 

\supplement{}

\acknowledgements{We thank N.\@ Treps, V.\@ Parigi, and especially M.\@ Walschaers for stimulating discussions.
We also thank A.\@ Leverrier for interesting discussion on de Finetti reductions, and useful comments on previous versions of this work. This work was supported by the ANR project ANR-13-BS04-0014 COMB.}

\nolinenumbers 

\EventEditors{Steven T. Flammia}
\EventNoEds{1}
\EventLongTitle{15th Conference on the Theory of Quantum Computation, Communication and Cryptography (TQC 2020)}
\EventShortTitle{TQC 2020}
\EventAcronym{TQC}
\EventYear{2020}
\EventDate{June 9--12, 2020}
\EventLocation{Riga, Latvia}
\EventLogo{}
\SeriesVolume{158}
\ArticleNo{3}

\hideLIPIcs

\begin{document}

\maketitle

\begin{abstract}
In this work we develop new methods for the characterisation of continuous variable quantum states using heterodyne measurement in both the trusted and untrusted settings.
First, building on quantum state tomography with heterodyne detection, we introduce a reliable method
for continuous variable quantum state certification, which directly yields the elements of the density
matrix of the state considered with analytical condence intervals. This method neither needs
mathematical reconstruction of the data nor discrete binning of the sample space, and uses a single
Gaussian measurement setting. Second, beyond quantum state tomography and without its identical copies
assumption, we promote our reliable tomography method to a general efficient protocol for verifying continuous
variable pure quantum states with Gaussian measurements against fully malicious adversaries, i.e.,
making no assumptions whatsoever on the state generated by the adversary. These results are
obtained using a new analytical estimator for the expected value of any operator acting on a
continuous variable quantum state with bounded support over the Fock basis, computed with samples
from heterodyne detection of the state. 
\end{abstract}

\section{Introduction}

\noindent Out of the many properties featured by quantum physics,
the impossibility to perfectly determine an unknown state~\cite{d1996impossibility}
is specially interesting. 
This property is at the heart of quantum cryptography protocols such as quantum key distribution~\cite{BB84}.
On the other hand, it makes certification of the correct functioning of 
quantum devices a challenge, 
since the output of such devices can only be determined approximately, 
through repeated measurements over numerous copies of the output states.
With rapidly developing quantum technologies for communication, simulation, computation and sensing, 
the ability to assess the correct functioning of quantum devices is of major importance, for near-term systems, the so-called Noisy Intermediate-Scale Quantum (NISQ) 
devices~\cite{preskill2018quantum}, and for the more sophisticated devices. 

Depending on the desired level of trust, various methods are available for certifying the output of quantum devices. In the following, the task of checking the output state of a quantum device 
is denoted \textit{tomography} for state independent methods, when i.i.d.\@ behaviour is assumed, \textit{certification} for a given a target state, when i.i.d.\@ behaviour is assumed,
and \textit{verification} for a given target state, with no assumption whatsoever, and in particular without the i.i.d.\@ assumption. 

Quantum state tomography~\cite{d2003quantum} is an important technique which aims at reconstructing a good approximation of the output state of a quantum device by performing multiple rounds of measurements on several copies of said output states. 
Given an ensemble of identically prepared systems, with measurement outcomes from the same
observable, one can build up a histogram, from which a probability density can be estimated.
According to Born's rule, this probability density is the square modulus of the state
coefficients, taken in the basis corresponding to the measurement. 
However, a single measurement setting cannot yield the full state information 
since the phase of its coefficients are then lost. 
Many sets of measurements on many subensembles must be performed and combined to reconstruct the density matrix of the state. 
The data do not yield the state directly, but rather indirectly through data analysis.
Quantum state tomography assumes an \textit{independent and identically distributed} (i.i.d.) behaviour for the device, i.e., that the density matrix of the output state considered is the same at each round of measurement. This assumption may be relaxed with a tradeoff in the efficiency of the protocol~\cite{christandl2012reliable}. 

A certification task corresponds to a setting where one wants to benchmark an industrial quantum device, or check the output of a physical experiment. On the other hand, a verification task corresponds to a cryptographic scenario, where the device to be tested is untrusted, or the quantum data is given by a potentially malicious party, for example in the context of delegated quantum computing.
In the latter case, the task of quantum verification is to ensure that either the device behaved properly, or the computation aborts with high probability.
While delegated computing is a natural platform for the emerging NISQ devices, one can provide a physical interpretation to this adversarial setting by emphasising that we aim for deriving verification schemes that make no assumptions whatsoever about the noise model of the underlying systems.
Various methods for verification of quantum devices have been investigated, in particular for discrete variable quantum information~\cite{gheorghiu2017verification}, and they provide different efficiencies and security parameters depending on the computational power of the verifier. The common feature for all of these approaches is to utilise some basic obfuscation scheme that allows to reduce the problem of dealing with a fully general noise model, or a fully general adversarial deviation of the device, to a simple error detection scheme~\cite{vidick2018verification}.
 
In this work, we consider the setting of quantum information with continuous variables~\cite{lloyd1999quantum}, 
in which quantum states live in an infinite-dimensional Hilbert space. Using continuous variable systems for quantum computation and more general quantum information processing is a powerful alternative
to the discrete variable case.
\textit{Firstly}, it is compatible with standard network optics technology, where more efficient measurements are available. \textit{Secondly}, it allows for unprecedented scaling in entanglement, with entangled states of up to tens of thousands of subsystems reported~\cite{yokoyama2013ultra} generated deterministically. 

A continuous variable quantum process or state can be described by a quasi-probability
distribution in phase space, often the Wigner function~\cite{wigner1997quantum}, but 
also the Husimi $Q$ function or the Glauber--Sudarshan $P$ function~\cite{cahill1969density}. 
This allows for a simple and experimentally relevant classification of quantum states: 
those with a Gaussian quasiprobability distibutions are called Gaussian states, and the others non-Gaussian states. By extension, operations mapping Gaussian states to Gaussian states are also called Gaussian. 
These Gaussian operations and states are the ones implementable with linear optics and
quadratic non-linearities~\cite{Braunstein2005}, and are hence relatively easy to construct
experimentally.
However, it is well known that for many important applications, Gaussian operations and Gaussian states are not sufficient. This takes the forms of no-go theorems for distillation and error correction~\cite{eisert2002distilling,fiuravsek2002gaussian,niset2009no}, and the fact that all Gaussian computations can be simulated efficiently classically~\cite{Bartlett2002}. Furthermore, it is not possible to demonstrate non-locality or contextuality---which are increasingly understood to be important resources in quantum information---in the Gaussian regime.

For continuous variable quantum devices, checking that the output state is close to a target state may be done with linear optics using optical homodyne tomography~\cite{lvovsky2009continuous}. This method allows to reconstruct the Wigner function of a generic state using only Gaussian measurements, namely homodyne detection. 
Because of the continuous character of its outcomes,  
one must proceed to a discrete binning of the sample space, in order to build probability histograms. 
Then, the state representation in phase space is determined by a mathematical reconstruction. 

For cases where we have a specific target state, more efficient options are possible. 
For multimode Gaussian states, more efficient certification methods have been derived with Gaussian measurements~\cite{aolita2015reliable}. 
These methods involve the computation of a fidelity witness, i.e., a lower bound on the fidelity, from the measured samples. 
The cubic phase state certification protocol of~\cite{liu2019client} also introduces a fidelity witness and is an example of certification of a specific non-Gaussian state with Gaussian measurements, which assumes an i.i.d.\@ state preparation. 
The verification protocol for Gaussian continuous variable weighted hypergraph states of~\cite{takeuchi2018resource} removes this assumption, again for this specific family of states.


\section{Results}

\noindent In this work we address two main issues. \textit{Firstly}, existing continuous variable state tomography methods are not reliable in the sense of~\cite{christandl2012reliable}, because errors coming from the reconstruction procedure are indistinguishable from errors coming from the data. \textit{Secondly}, to the best of our knowledge there is no Gaussian verification protocol for non-Gaussian states without i.i.d.\@ assumption (a possible route using Serfling's bound was mentioned in Ref.~\cite{liu2019client} for removing the i.i.d.\@ assumption for their protocol). 

We thus introduce a general \textit{receive-and-measure} protocol for building trust for continuous variable quantum states, using solely Gaussian measurements, namely heterodyne detection~\cite{ferraro2005gaussian,teo2017heterodyning}. This protocol allows to perform reliable continuous variable quantum state tomography based on heterodyne detection, with analytical confidence intervals, which we refer to as \textit{heterodyne tomography} in what follows. This tomography technique only requires a single fixed measurement setting, compared to homodyne tomography. This protocol also provides a means for certifying continuous variable quantum states with an energy test, under the i.i.d.\@ assumption. Finally, the same protocol also allows to verify continuous variable states, without the i.i.d.\@ assumption. For these three applications, the measurements performed are the same. It is only the number of subsystems to be measured and the classical post-processing performed that differ from one application to another.

We detail the structure of the protocol in the following. We give an estimator for the expected value of any operator acting on a state with 
bounded support over the Fock basis (Theorem~\ref{thmain}) by deriving an approximate version of the optical equivalence theorem for antinormal ordering~\cite{cahill1969density}. The estimate is expressed as an expected value under heterodyne detection. Similar estimates have been obtained in the context of imperfect heterodyne detection~\cite{paris1996density,paris1996quantum}. We go beyond these works in different respects:
using this result, we introduce a reliable heterodyne tomography method and compute analytical bounds on its efficiency (Theorem~\ref{thQST}). 
We then derive a \textit{receive-and-measure} certification protocol (against i.i.d.\@ adversary) for continuous variable quantum states, with Gaussian measurements (Theorem~\ref{thi.i.d.}). We further promote this certification technique to a verification protocol against fully malicious adversary  (Theorem~\ref{thVUCVQC}), using a de Finetti reduction for infinite-dimensional systems~\cite{renner2009finetti}.


\section{Description of the protocol}
\label{sec:protocol}

\begin{figure}
	\begin{center}
		\includegraphics[width=2.6in]{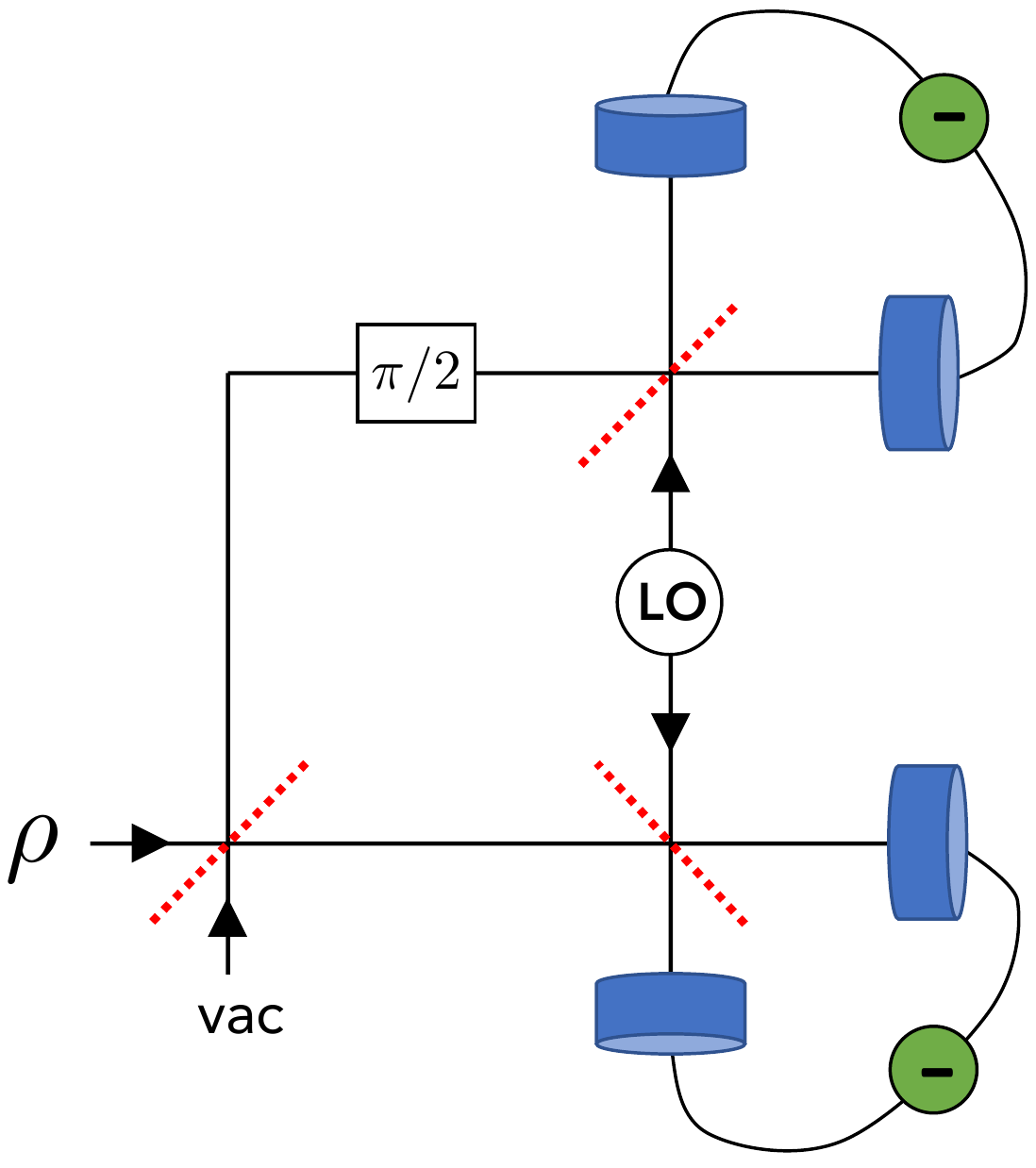}
		\caption{A schematic representation of heterodyne measurement of a state $\rho$. The dashed red lines represent balanced beamsplitters. \textit{LO} stands for local oscillator, i.e., strong coherent state, and \textit{vac} for vacuum state. The blue circles are photodiode detectors.}
		\label{fig:2homodyne}
	\end{center}
\end{figure}

\noindent Continuous variable quantum states live in an infinite-dimensional 
Hilbert space $\mathcal H$, spanned by the Fock basis $\{\ket n\}_{n\in\mathbb N}$, 
and are equivalently represented in phase space by their Husimi $Q$ function~\cite{cahill1969density}, a smoother relative of the Wigner function. Given a single-mode state $\rho$, its $Q$ function is defined as:
\be
Q_{\rho}(\alpha)=\frac{1}{\pi}\Tr\left(\ket\alpha\!\bra\alpha\rho\right)=\Tr\left(\Pi_\alpha\rho\right),
\ee
for all $\alpha\in\mathbb C$, where $\ket\alpha$ is a coherent state and where $\{\Pi_\alpha\}_{\alpha\in\mathbb C}=\left\{\frac{1}{\pi}\ket\alpha\!\bra\alpha\right\}_{\alpha\in\mathbb C}$ is the Positive Operator Valued Measure for heterodyne detection. 

This detection, also called double homodyne or eight-port homodyne~\cite{ferraro2005gaussian}, consists in splitting the measured state with a beamsplitter, and measuring both ends with homodyne detection (Fig.~\ref{fig:2homodyne}). This corresponds to a joint noisy measurement of quadratures $q$ and $p$. 
This is a Gaussian measurement, which yields two real outcomes, corresponding to the real and imaginary parts of $\alpha$. The $Q$ function of a single-mode state thus is a probability density function over $\mathbb C$ and measuring a state with heterodyne detection amounts to sampling from its $Q$ function.

Using this detection, one may acquire knowledge about an unknown continuous variable quantum state. More precisely, we define the following \textit{receive-and-measure} protocol, depicted in Fig.~\ref{fig:protocol}: given a quantum state $\rho^n$ over $n$ subsystems, measure some of the subsystems with heterodyne detection. Then, post-process the samples obtained to retrieve information about the remaining subsystems. The number subsystems to be measured and the post-processing performed depend on the application considered.

We show in the following sections how this protocol may be used to perform reliable tomography, certification and verification of continuous variable quantum states, and we detail the corresponding choices of subsystems and the classical post-processing for each task. 
\begin{figure}
	\begin{center}
		\includegraphics[width=3.1in]{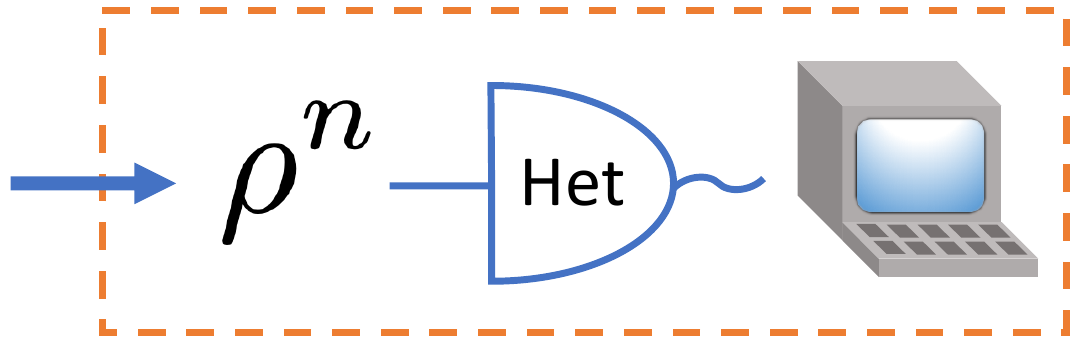}
		\caption{A schematic representation of the protocol. The tester (within the dashed rectangle) receives a continuous variable quantum state $\rho^n$ over $n$ subsystems. This state could be for example the outcome of $n$ successive runs of a physical experiment, the output of a commercial quantum device, or directly sent by some untrusted quantum server. The tester measures with heterodyne detection some of the subsystems of $\rho^n$, and uses the samples and efficient classical post-processing to deduce information about the remaining subsystems.}
		\label{fig:protocol}
	\end{center}
\end{figure}
%


\section{Heterodyne estimator}
\label{sec:main}

\noindent This section contains our main technical result, 
an estimator for the expected value of an operator acting on a state with 
bounded support over the Fock basis, from samples of heterodyne detection of the state. From this result, we derive various protocols in the following sections, ranging from tomography to state verification.

We denote by $\underset{\alpha\leftarrow D}{\mathbb E}[f(\alpha)]$ the expected value of a function $f$ for samples drawn from a distribution $D$.
Let us introduce for $k,l\ge0$ the polynomials
\be
\mathcal{L}_{k,l}(z)=e^{zz^*}\frac{(-1)^{k+l}}{\sqrt{k!}\sqrt{l!}}\frac{\partial^{k+l}}{\partial z^k\partial z^{*l}}e^{-zz^*},
\label{2DLmain}
\ee
for $z\in\mathbb C$, which are, up to a normalisation, the Laguerre $2$D polynomials, appearing in particular in the expressions of Wigner function of Fock states~\cite{wunsche1998laguerre}. 
For any operator $A=\sum_{k,l=0}^{+\infty}{A_{kl}\ket k\!\bra l}$ and all $E\in\mathbb N$, we define with these polynomials the function
\be
f_A(z,\eta)=\frac1\eta e^{\left(1-\frac{1}{\eta}\right)zz^*}\sum_{k,l=0}^{E}{\frac{A_{kl}}{\sqrt{\eta^{k+l}}}\mathcal{L}_{k,l}\left(\frac z{\sqrt{\eta}}\right)},
\label{fmain}
\ee
for all $z\in\mathbb C$, and all $0<\eta<1$. We omit the dependency in $E$ for brevity. The function $z\mapsto f_A(z,\eta)$,
being a polynomial multiplied by a converging Gaussian function, 
is bounded over $\mathbb C$.
With the same notations, we also define the following constant:
\be
K_A=\sum_{k,l=0}^E{|A_{kl}|\sqrt{(k+1)(l+1)}}.
\label{Kmain}
\ee
The optical equivalence theorem for antinormal ordering~\cite{cahill1969density} gives an equivalence between the expectation value of an operator in Hilbert space and the expectation value of its Glauber-Sudarshan $P$ function. The $P$ function is however highly singular in general and our results are based instead on the following approximate version of this equivalence when the $P$ function is replaced by the bounded function $f$:

\begin{theorem} 
Let $E\in\mathbb N$ and let \mbox{$0<\eta<\frac2{E}$}. Let also $A=\sum_{k,l=0}^{+\infty}{A_{kl}\ket k\!\bra l}$ be an operator and let \mbox{$\rho=\sum_{k,l=0}^{E}{\rho_{kl}\ket k\!\bra l}$} be a density operator with bounded support.
Then,
\be
\left|\Tr\left(A\rho\right)-\underset{\alpha\leftarrow Q_{\mathrlap\rho}}{\mathbb E}[f_A(\alpha,\eta)]\right|\le\eta K_A,
\label{th1}
\ee
where the function $f$ and the constant $K$ are defined in Eqs.~(\ref{f}) and (\ref{K}).
\label{thmain}
\end{theorem}

\noindent For all theorems, the proof techniques are given in section~\ref{sec:proof} and the detailed proofs may be found in the appendix. This result provides an estimator for the expected value of any operator $A$ acting on a
continuous variable state $\rho$ with bounded support over the Fock basis. 
This estimator is the expected value of a bounded function $f_A$ over samples drawn 
from the probability density corresponding to a Gaussian measurement of $\rho$, 
namely heterodyne detection. The optical equivalence theorem for antinormal ordering corresponds to the limit $\eta\to0$.
The right hand side of Eq.~(\ref{th1}) is an energy bound, which depends on the operator $A$, the value $E$ and the precision parameter $\eta$.

When the operator $A$ is the density matrix of a continuous variable pure state $\ket\Psi$, the previous estimator approximates the fidelity $F(\Psi,\rho)=\braket{\Psi|\rho|\Psi}$ between $\ket\Psi\!\bra\Psi$ and $\rho$. With the same notations:

\begin{corollary} 
Let $E\in\mathbb N$ and let \mbox{$0<\eta<\frac2{E}$}. Let also $\ket\Psi\!\bra\Psi=\sum_{k,l=0}^{+\infty}{\psi_k\psi_l^*\ket k\!\bra l}$ be a normalised pure state and let \mbox{$\rho=\sum_{k,l=0}^{E}{\rho_{kl}\ket k\!\bra l}$} be a density operator with bounded support.
Then,
\be
\ba
\left|F\left(\Psi,\rho\right)-\underset{\alpha\leftarrow Q_{\mathrlap\rho}}{\mathbb E}[f_\Psi(\alpha,\eta)]\right|&\le\eta K_\Psi\le\frac\eta2(E+1)(E+2),
\ea
\label{co1}
\ee
where the function $f_A$ and the constant $K_A$ are defined in Eqs.~(\ref{f}) and (\ref{K}), for $A=\ket\Psi\!\bra\Psi$.
\label{corofidelity}
\end{corollary}

\noindent This result provides an estimator for the fidelity between any target pure state $\ket\Psi$ and any continuous variable (mixed) state $\rho$ with bounded support over the Fock basis. This estimator is the expected value of a bounded function $f_\Psi$ over samples 
drawn from the probability density corresponding to a Gaussian measurement of $\rho$, namely heterodyne detection. 
The right hand side of Eq.~(\ref{co1}) is an energy bound, which may be refined depending on the expression of $\ket\Psi$. In particular, the second bound is independent of the target state $\ket\Psi$. The assumption of bounded support makes sense for tomography, but not necessarily in an adversarial setting. We will relax this condition for the certification and verification protocols in the following, and indeed estimate the energy bound from the heterodyne measurements. Errors in this estimation are taken into account in the confidence statements.

Given these results, one may choose a target pure state $\ket\Psi$, and measure with heterodyne detection various copies of the output (mixed) state $\rho$ of a quantum device with bounded support over the Fock basis. Then, using the samples obtained, one may estimate the expected value of $f_\Psi$, thus obtaining an estimate of the fidelity between the states $\ket\Psi\!\bra\Psi$ and $\rho$. Using this result, we introduce a reliable method for performing continuous variable quantum state tomography using heterodyne detection.


\section{Reliable continuous variable state tomography}
\label{sec:2homQST}

\noindent Continuous variable quantum state tomography methods usually make two assumptions: firstly that the measured states are independent identical copies (i.i.d.\@ assumption, for \textit{independently and identically distributed}), and secondly that the measured states have a bounded support over the Fock basis~\cite{lvovsky2009continuous}.
With the same assumptions, we present a reliable method for state tomography with heterodyne detection which has the advantage of providing analytical confidence intervals.
Our method directly provides estimates of the elements of the state density matrix, 
phase included. As such, neither mathematical reconstruction of the phase, nor binning of the sample space is needed, since the samples are used only to compute expected values of bounded functions. 
Moreover, only a single fixed Gaussian measurement setting is needed, namely heterodyne detection (Fig.~\ref{fig:2homodyne}).

For tomographic application, all copies of the state are measured. For $n\ge1$, let $\alpha_1,\dots,\alpha_n\in\mathbb C$ be samples from heterodyne detection
of $n$ copies of a quantum state $\rho$. For $\epsilon>0$ and $k,l\in\mathbb N$, we define
\be
\rho^\epsilon_{kl}=\frac{1}{n}\sum_{i=1}^n{f_{\ket l\!\bra k}\left(\alpha_i,\epsilon/K_{\ket l\!\bra k}\right)},
\label{Fkleps}
\ee
where the function $f_A$ and the constant $K_A$ are defined in Eqs.~(\ref{f}) and (\ref{K}), for $A=\ket l\!\bra k$, and where $\epsilon>0$ is a free parameter. The quantity $\rho^\epsilon_{kl}$ is the average of the function $f_{\ket l\!\bra k}$ over the samples $\alpha_1,\dots,\alpha_n$. The next result shows that this estimator approximates the matrix element $k,l$ of this state with high probability.
We use the notations of Theorem~\ref{thmain}.

\begin{theorem}[Reliable heterodyne tomography]\label{thQST}
Let $\epsilon,\epsilon'>0$, $n\ge1$ and $\alpha_1,\dots,\alpha_n$ be samples obtained by measuring with heterodyne detection $n$ copies of a state $\rho=\sum_{k,l=0}^E{\rho_{kl}\ket k\!\bra l}$ with bounded support, for $E\in\mathbb N$. Then
\be
\left|\rho_{kl}-\rho^\epsilon_{kl}\right|\le\epsilon+\epsilon',
\label{QST}
\ee
for all $0\le k, l\le E$, with probability greater than
\be
1-4\smashoperator{\sum_{0\le k \le l \le E}}\exp\left[{-\frac{n\epsilon^{2+k+l}\epsilon'^2}{4C_{kl}}}\right],
\ee
where the estimate $\rho^\epsilon_{kl}$ is defined in Eq.~(\ref{Fkleps}) and where
\be
C_{kl}=\left[(k+1)(l+1)\right]^{1+\frac{k+l}2}2^{|l-k|}\binom{\max{(k,l)}}{\min{(k,l)}}
\label{Ckl}
\ee
is a constant independent of $\rho$.
\end{theorem}

\noindent In light of this result, the principle for performing reliable heterodyne tomography is straightforward and as follows: $n$ identical copies $\rho^{\otimes n}$ of the output quantum state of a physical experiment or quantum device are measured with heterodyne detection, yielding the values $\alpha_1,\dots,\alpha_n$. These values are used to compute the estimates $\rho^\epsilon_{kl}$, defined in Eq.~(\ref{Fkleps}), for all $k,l$ in the range of energy of the experiment. Then, Theorem~\ref{thQST} directly provides confidence intervals for all these estimates of $\rho_{kl}$, the matrix elements of the density operator $\rho$, without the need for a binning of the sample space or any additional data reconstruction, using a single measurement setting. For a desired precision $\epsilon$ and a failure probability $\delta$, the number of samples needed scales as $n=\poly(1/\epsilon,\log(1/\delta))$.

Both homodyne and heterodyne quantum state tomography assume a bounded support over the Fock basis for the output state considered, i.e., that all matrix elements are equal to zero beyond a certain value, and that the output quantum states are i.i.d.\@, i.e., that all measured output states are independent and identical. While these assumptions are natural when looking at the output of a physical experiment, corresponding to a noisy partially trusted quantum device with bounded energy, they may be questionable in the context of untrusted devices. We remove these assumptions in what follows: we first drop the bounded support assumption, deriving a certification protocol for continuous variable quantum states of an i.i.d.\@ device with heterodyne detection ; then, we drop both assumptions, deriving a general verification protocol for continuous variable quantum states against an adversary who can potentially be fully malicious.


\section{State certification with Gaussian measurements}
\label{sec:Vi.i.d.}

\noindent Given an untrusted source of quantum states, the purpose of state certification and state verification protocols is to check whether if its output state is close to a given target state, or far from it. To achieve this, a verifier tests the output state of the source. Ideally, one would like to obtain an upper bound on the probability that the state is not close from the target state, given that it passed a test. However, this is known to be impossible without prior knowledge of the tested state distribution~\cite{gheorghiu2017verification}. Indeed, writing this conditional probability
\be
\Prb[\text{incorrect}|\text{accept}]=\frac{\Prb[\text{incorrect}\cap\text{accept}]}{\Prb[\text{accept}]},
\ee
in a situation where the device always produces a bad output state, it is rejected by the verifier's test most of the time, so the acceptance probability is very small and the conditional probability is equal to $1$. Therefore, the quantity that will always be bounded in certification and verification protocols, in which one does not have prior knowledge of the device, is the joint probability that the tested state is not close to the target state \textit{and} that it passes the test. Equivalently, we obtain lower bounds on the probability that the tested state is close to the target state or that it fails the test.

We first consider the certification of the output of an i.i.d.\@ quantum device, i.e., which output state is the same at each round. However, we do not assume that the output states of the device have bounded support over the Fock basis anymore. This is instead ensured probabilistically using the samples from heterodyne detection.

Our continuous variable quantum state certification protocol is then as follows: let $\ket\Psi$ be a target pure state, of which one wants to certify $m$ copies. The values $s$ and $E$ are free parameters of the protocol. One instructs the i.i.d.\@ device to prepare $n+m$ copies of $\ket\Psi$, and the device outputs an i.i.d.\@ (mixed) state $\rho^{\otimes(n+m)}$. One keeps $m$ copies $\rho^{\otimes m}$, and measures the $n$ others with heterodyne detection, obtaining the samples $\alpha_1,\dots,\alpha_n$. One records the number $r$ of samples such that $|\alpha_i|^2>E$. We refer to this step as \textit{support estimation}. For a given $\epsilon>0$, one also computes with the same samples the estimate
\be
F_\Psi(\rho)=\left[\frac{1}{n}\sum_{i=1}^n{f_\Psi\left(\alpha_i,\epsilon/(mK_\Psi)\right)}\right]^m,
\label{tildeF}
\ee
where the function $f_A$ and the constant $K_A$ are defined in Eqs.~(\ref{f}) and (\ref{K}), for $A=\ket\Psi\!\bra\Psi$, and where $\epsilon>0$ is a free parameter. The next result quantifies how close this estimate is from the fidelity between the remaining $m$ copies of the output state $\rho^{\otimes m}$ of the tested device and $m$ copies of the target state $\ket\Psi\!\bra\Psi^{\otimes m}$.

\begin{theorem}[Gaussian certification of continuous variable quantum states]\label{thi.i.d.}
Let $\epsilon,\epsilon'>0$, let $s\le n$, and let $\alpha_1,\dots,\alpha_n$ be samples obtained by measuring with heterodyne detection $n$ copies of a state $\rho$. Let $E$ in $\mathbb N$, and let $r$ be the number of samples such that $|\alpha_i|^2>E$. Let also $\ket\Psi$ be a pure state. Then for all $m\in\mathbb N^*$,
\be
\left|F(\Psi^{\otimes m},\rho^{\otimes m})-F_\Psi(\rho)\right|\le\epsilon+\epsilon',
\ee
or $r>s$, with probability greater than
\be
1-\left(P_{\text{support}}^{iid}+P_{\text{Hoeffding}}^{iid}\right),
\ee
\vspace{-5pt}
where
\vspace{-5pt}
\be
P_{\text{support}}^{iid}=\frac{(s+1)^{3/2}}n\exp\left[{\frac{(s+1)^2}{n+1}}\right],
\ee
\vspace{-10pt}
\be
P_{\text{Hoeffding}}^{iid}=2\exp\left[{-\frac{n\epsilon^{2+2E}\epsilon'^2}{2m^{4+2E}C^2_\Psi}}\right],
\ee
where the estimate $F_\Psi(\rho)$ is defined in Eq.~(\ref{tildeF}), and where
\be
C_\Psi=\sum_{k,l=0}^E{|\psi_k\psi_l|\left(\frac\epsilon m\right)^{E-\frac{k+l}2}K_\psi^{1+\frac{k+l}2}\sqrt{2^{|l-k|}\binom{\max{(k,l)}}{\min{(k,l)}}}}
\label{Cpsi}
\ee
is a constant independent of $\rho$, with the constant $K$ defined in Eq.~(\ref{Kmain}).
\end{theorem}

\noindent This results implies that the quantity $F_\Psi(\rho)$ is a good estimate of the fidelity $F(\Psi^{\otimes m},\rho^{\otimes m})$, or the score at the support estimation step is higher than $s$, with high probability.
The values of the energy parameters $E$ and $s$ should be chosen to guarantee completeness, i.e., that if the correct state $\ket\Psi$ is sent, then $r\le s$ with high probability. 

This theorem is valid for all continuous variable target pure states $\ket\Psi$, and the failure probability may be greatly reduced depending on the expression of $\ket\Psi$.
The number of samples needed for certifying a given number of copies $m$ with a precision $\epsilon$ and a failure probability $\delta$ scales as \mbox{$n=\text{poly}(m,1/\epsilon,1/\delta)$}. Note that the same protocol may be used to obtain reliable estimates of $\Tr(A\rho)$ for any operator $A$ under the i.i.d.\@ assumption, by setting $m=1$ and replacing $\Psi$ by $A$ in Eq.~(\ref{tildeF}).

This certification protocol is promoted to a verification protocol in the following section, by removing the i.i.d.\@ assumption.


\section{State verification with Gaussian measurements}
\label{sec:Vgeneral}

\noindent We now consider an adversarial setting, where a verifier delegates the preparation of a continuous variable quantum state to a potentially malicious party, called the \textit{prover}. One could see the verifier as the experimentalist in the laboratory and the prover as the noisy device, where we aim not to make any assumptions about its correct functionality or noise model. Given the absence of any direct error correction mechanism that permits a fault tolerant run of the device, the aim of verification is to ensure that a wrong outcome is not being accepted. In the context of state verification, this amounts to making sure that the output state of the tested device is close to an ideal target state. 

The prover is not supposed to have i.i.d.\@ behaviour. In particular, when asked for various
copies of the same state, the prover may actually send a large state entangled over all subsystems, 
possibly also entangled with a quantum system on his side.
In that case, the certification protocol derived in the previous section is not reliable. With usual tomography measurements, the number of samples needed for a given precision of the fidelity estimate scales exponentially in the number of copies to verify. This is an essential limitation of quantum tomography techniques, because they check all possible correlations between the different subsystems. 

However we prove that, because of the symmetry of the protocol, the verifier can assume that the prover is sending permutation-invariant states, i.e., states that are invariant under any permutation of their subsystems. With a specific support estimation step, reduced states of permutation-invariant states are close to mixture almost-i.i.d.\@ states, i.e., states that are i.i.d.\@ on almost all subsystems. At the heart of this reduction is the de Finetti theorem for infinite-dimensional systems of~\cite{renner2009finetti}, which allows restricting to an almost-i.i.d.\@ prover. 

Our verification protocol is then as follows: the verifier wants to verify $m$ copies of a target pure state $\ket\Psi$. The values $n$, $k$, $q$, $s$ and $E$ are free parameters of the protocol. The prover is instructed to prepare $n+k$ copies of $\ket\Psi$ and send them to the verifier. The verifier picks $k$ subsystems at random and measures them with heterodyne detection, obtaining the samples $\beta_1,\dots,\beta_k$, and records the number $r$ of values $|\beta_i|^2>E$. The verifier discards $4q$ subsystems at random and measures all the others but $m$ chosen at random with heterodyne detection, obtaining the samples $\alpha_1,\dots,\alpha_{n-4q-m}$. Finally, the verifier computes with these samples the estimate
\be
F_\Psi(\rho)=\left[\frac{1}{n-4q-m}\smashoperator{\sum_{i=1}^{n-4q-m}}{f_\Psi\left(\alpha_i,\epsilon/(mK_\Psi)\right)}\right]^m,
\label{tildeF2}
\ee
where the function $f_A$ and the constant $K_A$ are defined in Eqs.~(\ref{f}) and (\ref{K}), for $A=\ket\Psi\!\bra\Psi$ and where $\epsilon>0$ is a free parameter. Note that this estimate is identical to the one defined in Eq.~(\ref{tildeF}), replacing $n$ by $n-4q-m$.
 
\begin{theorem}[Gaussian verification of continuous variable quantum states] 
Let $n\ge1$, let $s\le k$, and let $\rho^{n+k}$ be a state over $n+k$ subsystems. Let $\beta_1,\dots,\beta_k$ be samples obtained by measuring $k$ subsystems at random with heterodyne detection and let $\rho^n$ be the remaining state after the measurement. Let $E$ in $\mathbb N$, and let $r$ be the number of samples such that $|\beta_i|^2>E$. Let also $q\ge m$, and let $\rho^m$ be the state remaining after discarding $4q$ subsystems of $\rho^n$ at random, and measuring $n-4q-m$ other subsystems at random with heterodyne detection, yielding the samples $\alpha_1,\dots,\alpha_{n-4q-m}$. Let $\epsilon,\epsilon'>0$ and let $\ket\Psi$ be a target pure state. Then,
\be
\left|F\left(\Psi^{\otimes m},\rho^m\right)-F_\Psi(\rho)\right|\le\epsilon+\epsilon'+P_{deFinetti},
\ee
or $r>s$, with probability greater than
\be
1-\left(P_{\text{support}}+P_{\text{deFinetti}}+P_{\text{choice}}+P_{\text{Hoeffding}}\right),
\ee
where
\be
P_{\text{support}}=8k^{3/2}\exp\left[{-\frac{k}{9}\left(\frac{q}{n}-\frac{2s}{k}\right)^2}\right],
\ee
\be
P_{\text{deFinetti}}=q^{(E+1)^2/2}\exp\left[{-\frac{2q(q+1)}{n}}\right],
\ee
\be
P_{\text{choice}}=\frac{m(4q+m-1)}{n-4q},
\ee
\be
P_{\text{Hoeffding}}=2\binom{n-4q}{4q}\exp\left[{-\frac{n-8q}{2m^{4+2E}}\left(\frac{\epsilon^{1+E}\epsilon'}{C_\Psi}-\frac{8qm^{2+E}}{n-4q-m}\right)^2}\right],
\ee
where the estimate $F_\Psi(\rho)$ is defined in Eq.~(\ref{tildeF2}), and where $C_\Psi$ is a constant independent of $\rho$ defined in Eq.~(\ref{Cpsi}).
\label{thVUCVQC}
\end{theorem}

\noindent This result implies that the quantity $F_\Psi(\rho)$ is a good estimate of the fidelity $F(\Psi^{\otimes m},\rho^m)$, or the score at the support estimation step is higher than $s$, with high probability.
Like for the certification protocol, the values of the energy parameters $E$ and $s$ should be chosen by the verifier to guarantee completeness, i.e., that if the prover sends the correct state $\ket\Psi$, then $r\le s$ with high probability. 

For specific choices of the free parameters of the protocol either the estimate $F_\Psi(\rho)$ is polynomially precise in $m$, or $r>s$, with polynomial probability in $m$, with $n,k,q=\text{poly}(m)$.
In particular, the efficiency of the protocol may be greatly refined by taking into account the expression of $\ket\Psi$ in the Fock basis, and optimizing over the free parameters.

This verification protocol let the verifier gain confidence about the precision of the estimate of the fidelity in Eq.~(\ref{tildeF2}). If the value of the estimate is close enough to $1$, the verifier may decide to use the state to run a computation. Indeed, statements on the fidelity of a state allow inferring the correctness of any trusted computation done afterwards using this state.
Let $\beta>0$, and let $\mathcal O$ be the observable corresponding to the result
of the trusted computation performed on $\rho^m$, the reduced state over $m$ subsystems
 instead of $\ket\Psi^{\otimes m}$, $m$ copies of the target state $\ket\Psi$. In other words, $\mathcal O$ encodes the resources 
which the verifier can perform perfectly (ancillary states, evolution and measurements),
the imperfections being encoded in $\rho$.
Then, $F\left(\Psi^{\otimes m},\rho^m\right)\ge1-\beta$ implies the following bound
on the total variation distance between the probability densities of the computation
output of the actual and the target computations:
\be
  \|P_{\Psi^{\otimes m}}^{\mathcal O}-P_{\rho^m}^{\mathcal O}\|_{tvd}
    \le D(\Psi^{\otimes m},\rho^m)
    \le\sqrt\beta,
\ee 
by standard properties of the trace distance $D$~\cite{fuchs1999cryptographic}. What this means is that the distribution of outcomes for the state $\rho^m$ sent by the prover 
is almost indistinguishable from the distribution of outcomes for $m$ copies of the ideal
state $\ket\Psi$, when the fidelity is close enough to one.


\section{Discussion}

\noindent Determining an unknown continuous variable quantum state is especially difficult since it is described by possibly infinitely many complex parameters. Existing methods like homodyne quantum state tomography require many different measurement settings, and heavy classical post-processing. For that purpose, we have introduced a reliable method for heterodyne quantum state tomography, which uses heterodyne detection as a single Gaussian measurement setting, and allows the retrieval of the density matrix of an unknown quantum state without the need for data reconstruction nor binning of the sample space. 
For data reconstruction methods such as Maximum Likelihood, errors from the reconstruction procedure are usually indistinguishable from errors coming from the tested quantum device. For that reason, such methods do not extend well to the task of verification, unlike our method.

Building on these tomography techniques, and with the addition of cryptographic techniques such as the de Finetti theorem, we have derived a protocol for verifying various copies of a continuous variable quantum state, without i.i.d.\@ assumption, with Gaussian measurements. This protocol is robust, as it directly gives a confidence interval on an estimate of the fidelity between the tested state and the target pure state. We emphasize that, while the target state is pure, the tested state is not required to be pure. 

Our verification protocol is complementary to the approach of~\cite{takeuchi2018resource}, in which a measurement-only verifier performs continuous variable quantum computing by delegating the preparation of Gaussian cluster states to a prover, and has to perform non-Gaussian measurements. In our approach, the measurement-only verifier may perform continuous variable quantum computing by delegating the preparation of non-Gaussian states to the prover, and has to perform Gaussian measurement, which are much easier to perform experimentally.

Our protocol may be tailored to different uses and assumptions, from tomography to verification, simply by changing the classical post-processing.
We expect this protocol to be useful for the validation of continuous variable 
quantum devices in the NISQ~\cite{preskill2018quantum} era and onwards.

In particular, an interesting perspective would be fine-tuning the various parameters of the protocol for specific target states in order to optimise its efficiency, thus reducing the number of samples needed for a given confidence interval. 
Another interesting prospect would be extending our main technical result, Theorem~\ref{thmain}, which applies to operators, to quantum maps. Also, in the case where the operator is the density matrix of a target pure state, our result provide an estimate for the fidelity, and it would be interesting to extend this to target mixed states.


\section{Proof techniques}
\label{sec:proof}

\noindent This section contains the primary mathematical tools used in the 
proofs of the theorems, along with some intuition. The full technical proofs are
detailed in the appendix.

The function $z\mapsto f_A(z,\eta)$ defined in Eq.~(\ref{fmain}) for $\eta>0$ is a bounded approximation of the Glauber-Sudarshan function $P_A$ of the operator $A$. This approximation is parametrised by a precision $\eta$, and a cutoff value $E$. The optical equivalence theorem for antinormal ordering~\cite{cahill1969density} reads
\be
\Tr(A\rho)=\int{Q_\rho(\alpha)P_A(\alpha)d^2\alpha}.
\ee
Given that
\be
\underset{\alpha\leftarrow Q_{\mathrlap\rho}}{\mathbb E}[f_A(\alpha,\eta)]=\int{Q_\rho(\alpha)f_A(\alpha,\eta)d^2\alpha},
\ee
we can expect that $\underset{\alpha\leftarrow Q_{\mathrlap{\rho}}}{\mathbb E}[f_A(\alpha,\eta)]$ is an approximation of $\Tr(A\rho)$ parametrised by $\eta$ and $E$. Theorem~\ref{thmain} makes this statement precise. We refer the reader to appendix~\ref{app:proofmaingen} for a detailed proof.

The proof of Theorem~\ref{thQST} combines Theorem~\ref{thmain} with Hoeffding inequality~\cite{hoeffding1963probability}, which quantifies the speed of convergence of the sample mean towards the expected value of a bounded i.i.d.\@ random variable:

\begin{lemma}\textbf{(Hoeffding)} Let $\lambda>0$, let $n\ge1$, let $z_1,\dots,z_n$ be i.i.d.\@ complex random variables from a probability density $D$ over $\mathbb R$, and let $f:\mathbb C\mapsto\mathbb R$ such that $|f(z)|\le M$, for $M>0$ and all $z\in\mathbb C$. Then
\be
\Prb\left[\left|\frac{1}{n}\sum_{i=1}^n{f(z_i)}- \underset{z\leftarrow D}{\mathbb E}[f(z)]\right|\ge\lambda\right] \leq 2\exp\left[{-\frac{n\lambda^2}{2M^2}}\right].
\ee
\end{lemma}
\noindent The proof then follows by applying this inequality for $D=Q_\rho$, and $f=f_{\ket k\bra l}$, for all values of $k,l$ between $0$ and $E$, together with the union bound. We refer the reader appendix~\ref{app:proofQST} for a detailed proof.

Theorem~\ref{thi.i.d.} removes the bounded support assumption and its proof is similar to the one of Theorem~\ref{thQST}, with the addition of a support estimation step, using samples from heterodyne detection. The main result utilised here is the fact that for all $E$~\cite{leverrier2013security}
\be
1-\Pi_{\le E}=\smashoperator{\sum_{n=E+1}^{+\infty}}{\ket n\bra n}\le\frac2\pi\smashoperator{\int_{\quad|\alpha|^2\ge E}}{\ket\alpha\bra\alpha d^2\alpha},
\ee
where $\Pi_{\le E}$ is the projector onto the space of states of support bounded by $E$. This result allows to bound the probability of having a large support and obtaining a low score at the support estimation step. We refer the reader to appendix~\ref{app:certification} for a detailed proof.

The proof of Theorem~\ref{thVUCVQC} is the most technical. This proof combines three main ingredients: a support estimation step for permutation-invariant states using samples from heterodyne detection, the de Finetti reduction from~\cite{renner2009finetti}, and a refined version of Hoeffding inequality for superpositions of almost-i.i.d.\@ states under a product measurement. We refer the reader to appendix~\ref{app:verification} for a detailed proof.


\bibliography{bibliography}


\appendix


\,\newpage


\begin{center}
{\LARGE\textbf{Appendix}}\\
\end{center}

\noindent In what follows, we prove the theorems from the main text. We first introduce a few notations and technical results, and then provide the detailed proofs. Proofs of intermediate technical results are indicated by a vertical bar.

\section{Preliminary material}
\label{app:fidelity}

\noindent The fidelity between two states $\rho,\sigma$ is defined as~\cite{nielsen2002quantum}
\be
F(\rho,\sigma)=\Tr\left(\sqrt{\sqrt\sigma\rho\sqrt\sigma}\right).
\ee
When at least one of the two states is a pure state, this expression reduces to
\be
F\left(\Psi,\rho\right)=\braket{\Psi|\rho|\Psi}.
\label{fidepure}
\ee
We write the Schatten $1$-norm of a bounded operator $T$ as
\be
\|T\|_1=\Tr\left(\sqrt{T^\dag T}\right)=\Tr(|T|).
\ee
The fidelity is related to the trace distance $D(\rho,\sigma)=\frac{1}{2}\|\rho-\sigma\|_1=\frac{1}{2}\Tr(|\rho-\sigma|)$ by~\cite{fuchs1999cryptographic}
\be
1-\sqrt{F(\rho,\sigma)}\le D(\rho,\sigma)\le \sqrt{1-F(\rho,\sigma)}.
\label{td1}
\ee
The trace distance verifies
\be
D(\rho,\sigma)=\max_{\mathcal O}{\|P_\rho^{\mathcal O}-P_\sigma^{\mathcal O}\|_{tvd}},
\label{td2}
\ee
where $P_\rho^{\mathcal O}$ (resp.\@ $P_\sigma^{\mathcal O}$) is the probability density associated to measuring the observable $\mathcal O$ for the state $\rho$ (resp.\@ $\sigma$), and where the maximum of the total variation distance $\|\cdot\|_{tvd}$ is taken over all observables. 
We introduce the following result:

\begin{lem} Let $0<\beta<1$. Let $\rho_1,\rho_2$ be two states such that $F(\rho_1,\rho_2)>1-\beta$. Let $\ket\Phi$ be a pure state, then
\be
\left|F(\Phi,\rho_1)-F(\Phi,\rho_2)\right|\le D(\rho_1,\rho_2)\le\sqrt\beta.
\ee
\label{lem:fidelitytriangular}
\end{lem}

\begin{leftbar}
\begin{proof} Let us write $P^\Phi_{\rho_1}$ and $P^\Phi_{\rho_2}$ the probability distributions associated to the binary measurement $\{\ket\Phi\bra\Phi,I-\ket\Phi\bra\Phi\}$ of the states $\rho_1$ and $\rho_2$, respectively. Then, $P^\Phi_{\rho_1}(0)+P^\Phi_{\rho_1}(1)=P^\Phi_{\rho_2}(0)+P^\Phi_{\rho_2}(1)=1$, and
\be
\ba
\|P^\Phi_{\rho_1}-P^\Phi_{\rho_2}\|_{tvd}&=\frac12\left(|P^\Phi_{\rho_1}(0)-P^\Phi_{\rho_2}(0)|+|P^\Phi_{\rho_1}(1)-P^\Phi_{\rho_2}(1)|\right)\\
&=|P^\Phi_{\rho_1}(0)-P^\Phi_{\rho_2}(0)|.
\ea
\ee
Hence,
\be
\ba
\left|F(\Phi,\rho_1)-F(\Phi,\rho_2)\right|&=\left|\braket{\Phi|\rho_1|\Phi}-\braket{\Phi|\rho_2|\Phi}\right|\\
&=|P^\Phi_{\rho_1}(0)-P^\Phi_{\rho_2}(0)|\\
&=\|P^\Phi_{\rho_1}-P^\Phi_{\rho_2}\|_{tvd}\\
&\le D(\rho_1,\rho_2)\\
&\le\sqrt{1-F(\rho_1,\rho_2)}\\
&\le\sqrt\beta,
\ea
\ee
where we used Eqs.~(\ref{td1}) and (\ref{td2}).

\end{proof}
\end{leftbar}

\noindent We also introduce the following simple results :

\begin{lem} Let $\eta>0$ and $a,b\in[0,1]$ such that $|a-b|\le\eta$. Then for all $m\ge1$,
\be
\left|a^m-b^m\right|\le m|a-b|\le m\eta.
\ee
\label{lem:simplem}
\end{lem}

\begin{leftbar}
\begin{proof}
With the notations of the lemma,
\be
\ba
\left|a^m-b^m\right|&=|a-b|\left|\sum_{j=0}^{m-1}{a^jb^{m-j-1}}\right|\\
&\le m|a-b|\\
&\le m\eta.
\ea
\ee
\end{proof}
\end{leftbar}

\noindent Let $E\in\mathbb N$, we write $\mathcal H$ the single-mode infinite-dimensional separable Hilbert space with the countable Fock basis $\{\ket n\}_{n\in\mathbb N}$, and $\bar{\mathcal H}=\text{span}(\ket0,\dots\ket E)$ the Hilbert space of states with bounded support over the Fock basis of size at most $E+1$, which has dimension $E+1$. 

We denote by $\underset{\alpha\leftarrow D}{\mathbb E}[f(\alpha)]$ the expected value of a function $f$ for samples drawn from a distribution $D$. \\

\noindent Let us introduce for $k,l\ge0$ the polynomials
\be
\ba
\mathcal{L}_{k,l}(z)&:=e^{zz^*}\frac{(-1)^{k+l}}{\sqrt{k!}\sqrt{l!}}\frac{\partial^{k+l}}{\partial z^k\partial z^{*l}}e^{-zz^*}\\
&=\smashoperator{\sum_{p=0}^{\min{(k,l)}}}{\frac{\sqrt{k!}\sqrt{l!}(-1)^p}{p!(k-p)!(l-p)!}z^{l-p}z^{*k-p}},
\ea
\label{2DL}
\ee
for all $z\in\mathbb C$, which are, up to a normalisation, the Laguerre $2$D polynomials, appearing in particular in the expressions of Wigner function of Fock states~\cite{wunsche1998laguerre}. 
For any operator $A=\sum_{k,l=0}^{+\infty}{A_{kl}\ket k\bra l}$ and all $E\in\mathbb N$, we define from these polynomials the function
\be
f_A(z,\eta):=\frac1\eta e^{\left(1-\frac{1}{\eta}\right)zz^*}\sum_{k,l=0}^{E}{\frac{A_{kl}}{\sqrt{\eta^{k+l}}}\text{ }\mathcal{L}_{k,l}\left(\frac z{\sqrt{\eta}}\right)},
\label{f}
\ee
for all $z\in\mathbb C$, and all $0<\eta<1$. We omit the dependencies in $E$. The function $z\mapsto f_A(z,\eta)$,
being a polynomial multiplied by a converging Gaussian function, 
is bounded over $\mathbb C$. 
With the same notations, we also define the following constant:
\be
K_A:=\sum_{k,l=0}^E{|A_{kl}|\sqrt{(k+1)(l+1)}}.
\label{K}
\ee
%


\section{Proof of Theorem~\ref{thmain} and Corollary~\ref{corofidelity}}
\label{app:proofmaingen}

\noindent The function $f_A$ defined in Eq.~(\ref{f}) is a bounded approximation of the Glauber-Sudarshan function $P_A$ of the operator $A$. This approximation is parametrised by a precision $\eta$, and a cutoff value $E$. The optical equivalence theorem for antinormal ordering~\cite{cahill1969density} reads
\be
\Tr(A\rho)=\int{Q_\rho(\alpha)P_A(\alpha)d^2\alpha}.
\ee
Given that
\be
\underset{\alpha\leftarrow Q_{\mathrlap\rho}}{\mathbb E}[f_A(\alpha,\eta)]=\int{Q_\rho(\alpha)f_A(\alpha,\eta)d^2\alpha},
\ee
we can expect that $\underset{\alpha\leftarrow Q_{\mathrlap{\rho}}}{\mathbb E}[f_A(\alpha,\eta)]$ is an approximation of $\Tr(A\rho)$ parametrised by $\eta$ and $E$. Theorem~\ref{thmain} makes this statement more precise, and we prove it in the following.

\subsection{Proof of Theorem~\ref{thmain}}
\label{app:proofmain}

\noindent We recall Theorem~\ref{thmain} from the main text:

\begin{theo} 
Let $E\in\mathbb N$ and let \mbox{$0<\eta<\frac2{E}$}. Let also $A=\sum_{k,l=0}^{+\infty}{A_{kl}\ket k\bra l}$ be an operator and let \mbox{$\rho=\sum_{k,l=0}^{E}{\rho_{kl}\ket k\bra l}$} be a density operator with bounded support.
Then,
\be
\left|\Tr\left(A\rho\right)-\underset{\alpha\leftarrow Q_{\mathrlap\rho}}{\mathbb E}[f_A(\alpha,\eta)]\right|\le\eta K_A,
\ee
where $f_A$ is defined in Eq.~(\ref{f}) and $K_A$ is defined in Eq.~(\ref{K}).
\end{theo}

\begin{proof}

With Eq.~(\ref{f}) we obtain
\be
\ba
\left|\Tr\left(A\rho\right)-\underset{\alpha\leftarrow Q_{\mathrlap{\rho}}}{\mathbb E}[f_A(\alpha,\eta)]\right|&=\left|\sum_{k,l=\mathrlap 0}^{+\infty}{A_{lk}\Tr\left(\ket l\bra k\rho\right)}-\smashoperator{\sum_{k,l=0}^E}{A_{lk}\underset{\alpha\leftarrow Q_{\mathrlap{\rho}}}{\mathbb E}[f_{\ket l\bra k}(\alpha,\eta)]}\right|\\
&=\left|\sum_{k,l=\mathrlap 0}^{E}{A_{lk}\left(\Tr\left(\ket l\bra k\rho\right)-\underset{\alpha\leftarrow Q_{\mathrlap{\rho}}}{\mathbb E}[f_{\ket l\bra k}(\alpha,\eta)]\right)}\right|\\
&\le\smashoperator{\sum_{k,l=0}^E}{\left|A_{lk}\right|\left|\Tr\left(\ket l\bra k\rho\right)-\underset{\alpha\leftarrow Q_{\mathrlap{\rho}}}{\mathbb E}[f_{\ket l\bra k}(\alpha,\eta)]\right|},
\ea
\label{app:triang}
\ee
where we used in the second line the fact that $\rho$ has a bounded support over the Fock basis. This shows that it is sufficient to prove the theorem for $A=\ket l\bra k$, for all $k,l$ from $0$ to $E$. We first introduce the following result:

\begin{lem} For all $0\le k,l\le E$,
\be
\underset{\alpha\leftarrow Q_{\mathrlap\rho}}{\mathbb E}[f_{\ket l\bra k}(\alpha,\eta)]=\rho_{kl}+\smashoperator{\sum_{\substack{m>k,n>l\\m-n=k-l}}^E}{\rho_{mn}\eta^{\frac{m+n-k-l}2}\sqrt{\binom mk\binom nl}}.
\ee
\label{lem:Eflk}
\end{lem}

\begin{leftbar}
\begin{proof}

Let us fix $k,l$ in $0,\dots,E$. By Eqs.~(\ref{2DL}) and (\ref{f}) we have, for all $z\in\mathbb C$,
\be
\ba
f_{\ket l\bra k}(z)&=\left(\frac{1}{\eta}\right)^{1+\frac{k+l}2}e^{\left(1-\frac{1}{\eta}\right)zz^*}
  \mathop{\mathcal{L}_{l,k}}\left(\frac{z}{\sqrt{\eta}}\right)\\
&=\left(\frac{1}{\eta}\right)^{1+\frac{k+l}2}e^{zz^*}\frac{(-1)^{\mathrlap{k+l}}}{\sqrt{k!}\sqrt{l!}}\left.\frac{\partial^{\mathrlap{k+l}}}{\partial u^{*k}\partial u^l}e^{-uu^*}\right\vert_{u=\frac{z}{\sqrt{\eta}}}\\
&=\frac{1}{\eta}e^{\left(1-\frac{1}{\eta}\right)zz^*}
   \sum_{p=0}^{\min{(\mathrlap{k,l)}}}{\frac{(-1)^p\sqrt{k!}\sqrt{l!}}{p!(k-p)!(l-p)!}
       \left(\frac1\eta\right)^{k+l-\mathrlap{p}}z^{k-p}z^{*l-p}}.
\ea
\ee
Moreover, for all $\alpha\in\mathbb C$,
\be
\ba
Q_\rho(\alpha)&=\frac1\pi\braket{\alpha|\rho|\alpha}\\
&=\frac1\pi\sum_{m,n=0}^E{\rho_{mn}\braket{\alpha|m}\braket{n|\alpha}}\\
&=\frac1\pi\sum_{m,n=0}^E{\rho_{mn}\frac{\alpha^{*m}\alpha^n}{\sqrt{m!n!}}}e^{-|\alpha|^2}.
\ea
\ee
Combining these expressions we obtain
\be
\ba
\underset{\alpha\leftarrow Q_{\mathrlap\rho}}{\mathbb E}[f_{\ket l\bra k}(\alpha,\eta)]
  &=\int{Q_\rho(\alpha)f_{\ket l\bra k}(\alpha,\eta)d^2\alpha}\\
  &=\frac{1}{\pi\eta}\sum_{m,n\mathrlap{=0}}^E\rho_{mn}\frac{\sqrt{k!}\sqrt{l!}}{\sqrt{m!}\sqrt{n!}}
    \sum_{p=0}^{\min{(\mathrlap{k,l)}}}\frac{(-1)^p}{p!(k-p)!(l-p)!}
    \left(\frac1\eta\right)^{k+l-\mathrlap{p}}\\
    &\quad\quad\quad\quad\quad\times\int{\alpha^{k+n-p}\alpha^{*(l+m-p)}e^{-\frac1\eta|\alpha|^2}d^2\alpha}.
\ea
\ee
Setting $\alpha=re^{i\theta}$, we have $d^2\alpha=rdrd\theta$ and the integral on the last line may be computed as
\be
\ba
\int{\alpha^{k+n-p}\alpha^{*(l+m-p)}e^{-\frac1\eta|\alpha|^2}d^2\alpha}
  &=\smashoperator{\int_0^{+\infty}}{r^{k+l+m+n-2p+1}e^{-\frac{r^2}\eta}dr}
    \smashoperator{\int_0^{2\pi}}{e^{i(k+n-l-m)\theta}d\theta}\\
&=\begin{cases} \pi\left(\frac{k+l+m+n}{2}-p\right)!\eta^{\frac{k+l+m+n}{2}-p+1}&\text{ for }k-l=m-n,\\&\\0&\text{ for }k-l\neq m-n,\end{cases}
\ea
\ee
where we used $\int_0^{+\infty}{r^{2t+1}e^{-\frac{r^2}\eta}}=\frac12t!\eta^{t+1}$ for $t=\frac{k+l+m+n}{2}-p$, which is obtained directly by induction and integration by parts (note that for $k-l=m-n$, and $p\le\min(k,l)$, we have indeed $t\in\mathbb N$). Hence,
\be
\ba
\underset{\alpha\leftarrow Q_{\mathrlap\rho}}{\mathbb E}[f_{\ket l\bra k}(\alpha,\eta)]
  &=\smashoperator{\sum_{\substack{m,n=0\\m-n=k-l}}^E}
      {\rho_{mn}\frac{\sqrt{k!}\sqrt{l!}}{\sqrt{m!}\sqrt{n!}}
      \sum_{p=0}^{\min{(k,l\mathrlap{)}}}{\frac{(-1)^p\left(\frac{k+l+m+n}{2}-p\right)!}{p!(k-p)!(l-p)!}\eta^{\frac{m+n-k-l}2}}}\\
  &=\smashoperator{\sum_{\substack{m,n=0\\m-n=k-l}}^E}{\rho_{mn}\eta^{\frac{m+n-k-l}2}\frac{\left(\frac{k+l+m+n}{2}\right)!}{\sqrt{m!}\sqrt{n!}\sqrt{k!}\sqrt{l!}}
  \sum_{p=0}^{\min(\mathrlap{k,l)}}{(-1)^p\frac{\binom kp\binom lp}{\binom{\frac{k+l+m+n}{2}}p}}}.
\ea
\label{interE}
\ee
Now for $k\le l$ we have, for all $q\in\mathbb N$ (see, e.g., result 7.1 of~\cite{gould1972combinatorial}),
\be
\sum_{p=0}^k{(-1)^p\frac{\binom kp\binom lp}{\binom qp}}=\begin{cases}\frac{\binom{q-l}{k}}{\binom{q}{k}}&\text{ for }q\ge k+l,\\&\\0&\text{ for }q<k+l.\end{cases}
\ee
When $k\le l$, Eq~(\ref{interE}) thus yields
\be
\ba
\underset{\alpha\leftarrow Q_{\mathrlap\rho}}{\mathbb E}[f_{\ket l\bra k}(\alpha,\eta)]
  &=\smashoperator{\sum_{\substack{m,n=0\\m-n=k-l\\m+n\ge k+l}}^E}
   {\rho_{mn}\eta^{\frac{m+n-k-l}2}\frac{\left(\frac{k+l+m+n}{2}\right)!}{\sqrt{m!}\sqrt{n!}\sqrt{k!}\sqrt{l!}}\frac{\binom{\frac{k+l+m+n}{2}-l}{k}}{\binom{\frac{k+l+m+n}{2}}{k}}}\\
  &=\smashoperator{\sum_{\substack{m\ge k,n\ge l\\m-n=k-l}}^E}
   {\rho_{mn}\eta^{\frac{m+n-k-l}2}\frac1{\sqrt{m!}\sqrt{n!}\sqrt{k!}\sqrt{l!}}\frac{\left(\frac{k-l+m+n}{2}\right)!\left(\frac{-k+l+m+n}{2}\right)!}{\left(\frac{-k-l+m+n}{2}\right)!}}\\
  &=\smashoperator{\sum_{\substack{m\ge k,n\ge l\\m-n=k-l}}^E}
    {\rho_{mn}\eta^{\frac{m+n-k-l}2}\frac{\sqrt{m!}\sqrt{n!}}{\sqrt{k!}\sqrt{l!}\sqrt{(m-k)!}\sqrt{(n-l)!}}}\\
  &=\smashoperator{\sum_{\substack{m\ge k,n\ge l\\m-n=k-l}}^E}
    {\rho_{mn}\eta^{\frac{m+n-k-l}2}\sqrt{\binom mk\binom nl}},
\ea
\ee
where we used that within the summation $m-n=k-l$. This formula is also valid for $l\le k$, with the same reasoning. We finally obtain, for any $k,l$ in $0,\dots,E$
\be
\ba
\underset{\alpha\leftarrow Q_{\mathrlap\rho}}{\mathbb E}[f_{\ket l\bra k}(\alpha,\eta)]
  &=\smashoperator{\sum_{\substack{m\ge k,n\ge l\\m-n=k-l}}^E}{\rho_{mn}\eta^{\frac{m+n-k-l}2}\sqrt{\binom mk\binom nl}}\\
&=\rho_{kl}
  +\smashoperator{\sum_{\substack{m>k,n>l\\m-n=k-l}}^E}{\rho_{mn}\eta^{\frac{m+n-k-l}2}\sqrt{\binom mk\binom nl}}.
\ea
\ee

\end{proof}
\end{leftbar}

\noindent Using Lemma~\ref{lem:Eflk}, we obtain
\be
\ba
\left|\Tr(\ket l\bra k\rho)-\underset{\alpha\leftarrow Q_{\mathrlap\rho}}{\mathbb E}
  [f_{\ket l\bra k}(\alpha,\eta)]\right|
&=\left|\rho_{kl}-\underset{\alpha\leftarrow Q_{\mathrlap\rho}}{\mathbb E}[f_{\ket l\bra k}(\alpha,\eta)]\right|\\
&=\left|\phantom{{}_{m.}}\smashoperator{\sum_{\substack{m>k,n>l\\m-n=k-l}}^E}
  {\rho_{mn}\eta^{\frac{m+n-k-l}2}\sqrt{\binom mk\binom nl}}\right|\\
&\le\smashoperator{\sum_{\substack{m>k,n>l\\m-n=k-l}}^E}{|\rho_{mn}|\eta^{\frac{m+n-k-l}2}\sqrt{\binom mk\binom nl}}\\
&= \sum_{s=1}^{\mathllap E-\mathrlap{\max{(k,l)}}}
  {|\rho_{s+k,s+l}|\eta^s\sqrt{\binom{s+k}k\binom{s+l}l}}\\
&\le \sum_{s=1}^{\mathllap E-\max\mathrlap{(k,l)}}
  {\eta^s\sqrt{\binom{s+k}k\binom{s+l}l}\sqrt{\rho_{s+k,s+k}}\sqrt{\rho_{s+l,s+l}}},
\ea
\label{interE2}
\ee
where we set $s=m-k=n-l=\frac{m+n-k-l}2$ in the third line, and where we used $|\rho_{s+k,s+l}|\le\sqrt{\rho_{s+k,s+k}}\sqrt{\rho_{s+l,s+l}}$ in the last line, since $\rho$ is a positive semidefinite matrix.
In order to obtain an upper bound independent of $\rho$, we now show for all $s$ that $\eta^s\sqrt{\binom{s+k}k\binom{s+l}l}\le\eta\sqrt{(k+1)(l+1)}$ for $\eta\le\frac2{E}$.
For all $k,l$ in $0,\dots,E$ and for all $s$ in $2,\dots,E-\max({k,l})$, we have
\be
\frac{\sqrt{s+k}\sqrt{s+l}}{s}\le\frac E2.
\ee
This in turn implies that for all $s$ in $2,\dots,E-\max({k,l})$
\be
\ba
\eta^s\sqrt{\binom{s+k}k\binom{s+l}l}&=\eta\frac{\sqrt{(s+k)(s+l)}}s\eta^{s-1}\sqrt{\binom{s-1+k}k\binom{s-1+l}l}\\
&\le\frac{\eta E}2\eta^{s-1}\sqrt{\binom{s-1+k}k\binom{s-1+l}l}\\
&\le\eta^{s-1}\sqrt{\binom{s-1+k}k\binom{s-1+l}l},
\ea
\ee
since we assumed $\eta\le\frac2{E}$. Hence by induction, for all $s$ in $2,\dots,E-\max({k,l})$,
\be
\eta^s\sqrt{\binom{s+k}k\binom{s+l}l}\le\eta^1\sqrt{\binom{1+k}k\binom{1+l}l}=\eta\sqrt{(k+1)(l+1)}.
\ee
Combining this with Eq.~(\ref{interE2}) yields
\be
\ba
\left|\Tr(\ket l\bra k\rho)-\underset{\alpha\leftarrow Q_{\mathrlap\rho}}{\mathbb E}[f_{\ket l\bra k}(\alpha,\eta)]\right|
&\le\eta\sqrt{(k+1)(l+1)}
 \sum_{s=1}^{\mathllap E-\mathrlap{\max(k,l)}}
 {\sqrt{\rho_{s+k,s+k}}\sqrt{\rho_{s+l,s+l}}}\\
&\le\eta\sqrt{(k+1)(l+1)}\sqrt{\sum_{s=1}^{\mathllap E-\mathrlap{\max(k,l)}}
 {\rho_{s+k,s+k}}\sum_{s=1}^{\mathllap E-\mathrlap{\max(k,l)}}{\rho_{s+l,s+l}}}\\ 
&\le\eta\sqrt{(k+1)(l+1)},
\ea
\ee
for all $k,l$ in $0,\dots,E$, where we used Cauchy-Schwarz inequality and the fact that $\Tr(\rho)=1$. Together with Eq.~(\ref{app:triang}) we obtain
\be
\ba
\left|\Tr\left(A\rho\right)-\underset{\alpha\leftarrow Q_{\mathrlap\rho}}{\mathbb E}
 [f_A(\alpha,\eta)]\right|
&\le\eta\smashoperator{\sum_{k,l=0}^E}{\left|A_{kl}\right|\sqrt{(k+1)(l+1)}}\\
&=\eta K_A,
\ea
\ee
by Eq.~(\ref{K}).

\end{proof}

\subsection{Proof of Corollary~\ref{corofidelity}}
\label{app:proofcoro}

\noindent We recall Corollary~\ref{corofidelity} from the main text:

\begin{coro} 
Let $E\in\mathbb N$ and let \mbox{$0<\eta<\frac2{E}$}. Let also $\ket\Psi\bra\Psi=\sum_{k,l=0}^{+\infty}{\Psi_k\Psi_l^*\ket k\bra l}$ be a normalised pure state and let \mbox{$\rho=\sum_{k,l=0}^{E}{\rho_{kl}\ket k\bra l}$} be a density operator with bounded support.
Then,
\be
\ba
\left|F\left(\Psi,\rho\right)-\underset{\alpha\leftarrow Q_{\mathrlap\rho}}{\mathbb E}[f_\Psi(\alpha,\eta)]\right|&\le\eta K_\Psi\le\frac\eta2(E+1)(E+2),
\ea
\label{co1app}
\ee
where the function $f$ is defined in Eq.~(\ref{f}) and the constant $K$ in Eq.~(\ref{K}).
\end{coro}

\begin{proof}

In order to prove Corollary~\ref{corofidelity}, we apply Theorem~\ref{thmain} for $A=\ket\Psi\bra\Psi$ a pure state. We obtain
\be
\ba
\left|\braket{\Psi|\rho|\Psi}-\underset{\alpha\leftarrow Q_{\mathrlap\rho}}{\mathbb E}[f_{\Psi}(\alpha,\eta)]\right|&\le\eta K_\Psi\\
&=\eta\smashoperator{\sum_{k,l=0}^E}{|\psi_k\psi_l|\sqrt{(k+1)(l+1)}}\\
&=\eta\left(\smashoperator{\sum_{n=0}^E}{|\psi_n|\sqrt{n+1}}\right)^2\\
&\le\eta\sum_{n=0}^{E}{|\psi_n|^2}\sum_{n=0}^E{(n+1)}\\
&\le\frac\eta2(E+1)(E+2),
\ea
\ee
where we used Cauchy-Schwarz inequality, and $\sum_{n=0}^{E}{|\psi_n|^2}\le\Tr(\ket\Psi\bra\Psi)=1$. 
Since $\ket\Psi$ is a pure state, we have $F(\Psi,\rho)=\braket{\Psi|\rho|\Psi}$, which concludes the proof.

\end{proof}


\section{Proof of Theorem~\ref{thQST}}
\label{app:proofQST}

\noindent Let $\epsilon,\epsilon'>0$, let $n\ge1$, and let $\bm{\alpha}=\alpha_1,\dots,\alpha_n$ be samples obtained by measuring with heterodyne detection $n$ copies of a state $\rho=\sum_{k,l=0}^E{\rho_{kl}\ket k\bra l}$ with bounded support, for $E\in\mathbb N$. Let us define
\be
\rho^\epsilon_{kl}=\frac{1}{n}\sum_{i=1}^n{f_{\ket l\bra k}\left(\alpha_i,\frac\epsilon{K_{\ket l\bra k}}\right)},
\label{Fklepsapp}
\ee
where the function $f$ and the constant $K$ are defined in Eqs.~(\ref{f}) and (\ref{K}).
In this section, we prove the following result:

\begin{theo} For all $0\le k, l\le E$,
\be
\left|\rho_{kl}-\rho^\epsilon_{kl}\right|\le\epsilon+\epsilon',
\label{QSTapp}
\ee
with probability greater than $1-4\smashoperator{\sum_{0\le k \le l \le E}}\exp\left[{-\frac{N\epsilon^{2+k+l}\epsilon'^2}{4C_{kl}}}\right]$, where
\be
C_{kl}=\left[(k+1)(l+1)\right]^{1+\frac{k+l}2}2^{|l-k|}\binom{\max{(k,l)}}{\min{(k,l)}}.
\ee
\end{theo}

\noindent The law of large numbers ensures that the sample average from independently and identically distributed (i.i.d.\@) random variables converges to the expected value of these random variables, when the number of samples goes to infinity. The following lemma refines this statement and quantifies the speed of convergence:

\begin{lem}\textbf{(Hoeffding)} Let $\lambda>0$, let $n\ge1$, let $z_1,\dots,z_n$ be i.i.d.\@ complex random variables from a probability density $D$ over $\mathbb R$, and let $f:\mathbb C\mapsto\mathbb R$ such that $|f(z)|\le M$, for $M>0$ and all $z\in\mathbb C$. Then
\be
\Prb\left[\left|\frac{1}{n}\sum_{i=1}^n{f(z_i)}- \underset{z\leftarrow D}{\mathbb E}[f(z)]\right|\ge\lambda\right] \leq 2\exp\left[{-\frac{n\lambda^2}{2M^2}}\right].
\ee
\label{lem:HoeffdingR}
\end{lem}

\noindent This comes directly from Hoeffding inequality~\cite{hoeffding1963probability} applied to the real bounded i.i.d.\@ random variables $f(z_1),\dots,f(z_N)$.
When dealing with complex random variables, we use the following result instead:

\begin{lem}\textbf{(Hoeffding for complex random variables)} Let $\lambda>0$, let $n\ge1$, let $z_1,\dots,z_n$ be i.i.d.\@ complex random variables from a probability density $D$ over $\mathbb C$, and let $f:\mathbb C\mapsto\mathbb C$ such that $|f(z)|\le M$, for $M>0$ and all $z\in\mathbb C$. Then
\be
\Prb\left[\left|\frac{1}{n}\sum_{i=1}^n{f(z_i)}- \underset{z\leftarrow D}{\mathbb E}[f(z)]\right|\ge\lambda\right] \leq 4\exp\left[{-\frac{n\lambda^2}{4M^2}}\right].
\ee
\label{lem:HoeffdingC}
\end{lem}

\begin{leftbar}
\begin{proof} For all $a>0$ and all $z\in\mathbb C$, $|z|=\sqrt{\re{(z)}^2+\im{(z)}^2}\ge a$ implies $|\re{(z)}|\ge a/\sqrt2$ or $|\im{(z)}|\ge a/\sqrt2$. Hence,
\be
\Prb\left[|z|\ge a\right]\le\Prb\left[|\re{(z)}|\ge \frac{a}{\sqrt2}\right]+\Prb\left[|\im{(z)}|\ge \frac{a}{\sqrt2}\right],
\ee
so applying twice Lemma~\ref{lem:HoeffdingR} for the real random variables $\re{(f(z))}$ and $\im{(f(z))}$, respectively, yields Lemma~\ref{lem:HoeffdingC}.

\end{proof}
\end{leftbar}

\noindent In order to prove Theorem~\ref{thQST}, we apply Lemma~\ref{lem:HoeffdingC} to the functions $z\mapsto f_{\ket k\bra l}(z,\eta)$ defined in Eq.~(\ref{f}). We first need to bound them:

\begin{lem} For all $k,l\ge0$, define
\be
M_{kl}:=\sqrt{2^{|l-k|}\binom{\max{(k,l)}}{\min{(k,l)}}}.
\label{Mkl}
\ee
Then for all $k,l$ and all $z\in\mathbb C$,
\be
\left|f_{\ket k\bra l}(z,\eta)\right|\le\frac{M_{kl}}{\eta^{1+\frac{k+l}2}}.
\label{maxfkl}
\ee
\label{lem:boundfkl}
\end{lem}

\begin{leftbar}
\begin{proof} For $k$ or $l>E$ the inequality is trivial. For all $k,l\le E$ and all $z\in\mathbb C$,
\be
\ba
\left|f_{\ket k\bra l}(z,\eta)\right|&=\left(\frac{1}{\eta}\right)^{1+\frac{k+l}2}e^{\left(1-\frac{1}{\eta}\right)|z|^2}\left|\mathcal{L}_{k,l}\left(\frac{z}{\sqrt\eta}\right)\right|\\
&=\frac{1}{\eta}e^{\left(1-\frac{1}{\eta}\right)|z|^2}\frac{1}{\sqrt{k!}\sqrt{l!}}\left|\sum_{p=0}^{\min\mathrlap{(k,l)}}{\frac{(-1)^pk!l!}{p!(k-p)!(l-p)!}\frac1{\eta^{k+l-p}}z^{l-p}z^{*(k-p)}}\right|,
\ea
\label{maxfkl1}
\ee
where we used Eq.~(\ref{2DL}). Now for all $z\in\mathbb C^*$ and all $a>0$ we have~\cite{wunsche1998laguerre}
\be
\ba
\left|\sum_{p=0}^{\min\mathrlap{(k,l)}}{\frac{(-1)^pk!l!}{p!(k-p)!(l-p)!}a^{k+l-p}z^{l-p}z^{*(k-p)}}\right|&=a^kl!|z|^{k-l}\left|L_l^{(k-l)}\left(a|z|^2\right)\right|\\
&=a^lk!|z|^{l-k}\left|L_k^{(l-k)}\left(a|z|^2\right)\right|,
\ea
\ee
where
\be
L_n^{(\alpha)}(x)=\sum_{q=0}^n{\frac{(-1)^q}{q!}\binom{n+\alpha}{n-q}x^q}
\ee
are the generalised Laguerre polynomials~\cite{abramowitz1965handbook}, defined for $\alpha\in\mathbb R$ and $n\in\mathbb N$. Plugging this relation into Eq.~(\ref{maxfkl1}) we obtain
\be
\ba
\left|f_{\ket k\bra l}(z,\eta)\right|&=e^{\left(1-\frac{1}{\eta}\right)|z|^2}\frac{|z|^{l-k}}{\eta^{1+l}}\frac{\sqrt{k!}}{\sqrt{l!}}\left|L_k^{(l-k)}\left(\frac{|z|^2}{\eta}\right)\right|\\
&=e^{\left(1-\frac{1}{\eta}\right)|z|^2}\frac{|z|^{k-l}}{\eta^{1+k}}\frac{\sqrt{l!}}{\sqrt{k!}}\left|L_l^{(k-l)}\left(\frac{|z|^2}{\eta}\right)\right|,
\ea
\label{maxfkl2}
\ee
for all $z\in\mathbb C$. The generalised Laguerre polynomials are bounded as~\cite{rooney1985further}
\be
\left|L_n^{(\alpha)}(x)\right|\le\frac{\Gamma(n+\alpha+1)}{n!\Gamma(\alpha+1)}e^{\frac{x}{2}},
\label{boundLag1}
\ee
for all $x\ge0$, all $\alpha\ge0$ and all $n\in\mathbb N$, and as
\be
\left|L_n^{(\alpha)}(x)\right|\le2^{-\alpha}e^{\frac{x}{2}},
\label{boundLag2}
\ee
for all $x\ge0$, all $\alpha\le-\frac12$ and all $n\in\mathbb N$. 

Let $a>0$. Assuming $k<l$, we have $|z|^{l-k}\le a^{l-k}$ for $|z|\le a$, and $|z|^{k-l}\le a^{k-l}$ for $|z|\ge a$. Thus, the first line of Eq.~(\ref{maxfkl2}), together with Eq.~(\ref{boundLag1}), give
\be
\ba
\left|f_{\ket k\bra l}(z,\eta)\right|&\le e^{\left(1-\frac{1}{\eta}\right)|z|^2}\frac{a^{l-k}}{\eta^{1+l}}\frac{\sqrt{k!}}{\sqrt{l!}}\frac{l!}{k!(l-k)!}e^{\frac{|z|^2}{2\eta}}\\
&\le\frac{a^{l-k}}{\eta^{1+l}}\frac{\sqrt{l!}}{(l-k)!\sqrt{k!}},
\ea
\label{maxfkl3}
\ee
for $|z|\le a$ and $k<l$. Similarly, the second line of Eq.~(\ref{maxfkl2}), together with Eq.~(\ref{boundLag2}), give
\be
\ba
\left|f_{\ket k\bra l}(z,\eta)\right|&\le e^{\left(1-\frac{1}{\eta}\right)|z|^2}\frac{a^{k-l}}{\eta^{1+k}}\frac{\sqrt{l!}}{\sqrt{k!}}2^{l-k}e^{\frac{|z|^2}{2\eta}}\\
&\le\frac{a^{k-l}}{\eta^{1+k}}\frac{\sqrt{l!}}{\sqrt{k!}}2^{l-k},
\ea
\label{maxfkl4}
\ee
for $|z|\ge a$ and $k<l$. These two last bounds~(\ref{maxfkl3}, \ref{maxfkl4}) are equal for $a^{l-k}=(2\eta)^{\frac{l-k}{2}}\sqrt{(l-k)!}$, yielding the bound
\be
\left|f_{\ket k\bra l}(z,\eta)\right|\le\sqrt{\frac{2^{l-k}}{\eta^{2+k+l}}\binom{l}{k}},
\label{maxfklle}
\ee
for all $z\in\mathbb C$ and $k<l$. For $l<k$ the same reasoning gives
\be
\left|f_{\ket k\bra l}(z,\eta)\right|\le\sqrt{\frac{2^{k-l}}{\eta^{2+k+l}}\binom{k}{l}}.
\label{maxfklge}
\ee
Finally, for $k=l$ the previous bounds also hold, by combining Eqs.~(\ref{maxfkl2}) and (\ref{boundLag1}), and this proves Lemma~\ref{lem:boundfkl}.

\end{proof}
\end{leftbar}

\noindent Let $k,l\ge0$, $n\in\mathbb N$ and $\epsilon'>0$. Applying Lemma~\ref{lem:HoeffdingC} to the function $f_{\ket l\bra k}$, with the bound from Lemma~\ref{lem:boundfkl} yields
\be
\Prb\left[\left|\frac{1}{n}\sum_{i=1}^n{f_{\ket l\bra k}(\alpha_i,\eta)}- \underset{\alpha\leftarrow Q_{\mathrlap\rho}}{\mathbb E}[f_{\ket l\bra k}(\alpha,\eta)]\right|\ge\epsilon'\right] \leq 4\exp\left[{-\frac{n\eta^{2+k+l}\epsilon'^2}{4M_{kl}^2}}\right].
\label{Hoeffapp1}
\ee
Applying Theorem~\ref{thmain} for $A=\ket l\bra k$ we also obtain
\be
\left|\rho_{kl}-\underset{\alpha\leftarrow Q_{\mathrlap\rho}}{\mathbb E}\left[f_{\ket l\bra k}(\alpha,\eta)\right]\right|\le\eta\sqrt{k+1}\sqrt{l+1}.
\label{th1app}
\ee
Let $\alpha_1,\dots,\alpha_n$ be samples from the $Q$-function of $\rho$. Combining Eqs.~(\ref{Hoeffapp1}) and (\ref{th1app}), we obtain with the triangular inequality
\be
\left|\rho_{kl}-\frac{1}{n}\sum_{i=1}^n{f_{\ket l\bra k}(\alpha_i,\eta)}\right|\le\eta\sqrt{k+1}\sqrt{l+1}+\epsilon',
\ee
with probability greater than
\be
1-4\exp\left[{-\frac{n\eta^{2+k+l}\epsilon'^2}{4M_{kl}^2}}\right].
\ee
We have $K_{\ket k\bra l}=\sqrt{(k+1)(l+1)}$ by Eq.~(\ref{K}). Taking $\eta=\frac\epsilon{K_{\ket k\bra l}}$ yields
\be
\left|\rho_{kl}-\frac{1}{n}\sum_{i=1}^n{f_{\ket l\bra k}\left(\alpha_i,\frac{\epsilon}{K_{\ket k\bra l}}\right)}\right|\le\epsilon+\epsilon',
\ee
with probability greater than
\be
1-4\exp\left[{-\frac{n\epsilon^{2+k+l}\epsilon'^2}{4C_{kl}}}\right],
\label{probaQSTkl}
\ee
where we defined
\be
\ba
C_{kl}&:=\left[(k+1)(l+1)\right]^{1+\frac{k+l}2}M_{kl}^2\\
&=\left[(k+1)(l+1)\right]^{1+\frac{k+l}2}2^{|l-k|}\binom{\max{(k,l)}}{\min{(k,l)}}.
\ea
\label{Ckl}
\ee
Now this holds for $0\le k,l\le E$. Together with the union bound, this proves the theorem.

\qed


\section{Proof of Theorem~\ref{thi.i.d.}}
\label{app:certification}

\noindent In this section, we detail the proof of Theorem~\ref{thi.i.d.}, which goes along the same lines as the one of Theorem~\ref{thQST}, given in Section~\ref{app:proofQST}, with the addition of a support estimation step for i.i.d.\@ states.

\subsection{Support estimation for i.i.d.\@ states}
\label{app:Etesti.i.d.}

\noindent Let us define the following operators for $E\ge0$:
\be
U=\smashoperator{\sum_{n=E+1}^{+\infty}}{\ket n\bra n}=1-\Pi_{\le E},
\ee
where $\Pi_{\le E}=\sum_{n=0}^E{\ket n\bra n}$ is the projector onto the Hilbert space $\bar{\mathcal H}$ of states with less than $E$ photons, and
\be
T=\frac{1}{\pi}\smashoperator{\int_{\quad|\alpha|^2\ge E}}{\ket\alpha\bra\alpha d^2\alpha},
\ee
where $\ket\alpha$ is a coherent state. We have the following result, proven in~\cite{leverrier2013security} by expanding $T$ in the Fock basis:
\be
U\le2T.
\label{U2T}
\ee
In particular, 
\be
\mathrm{Tr}(U\rho)\le2\mathrm{Tr}(T\rho).
\label{ineqtrUT}
\ee
The probability $P_r$ that exactly $r$ among $n$ values of $|\alpha_i|^2$ are bigger than $E$ and $n-r$ values are lower, and that the projection $\Pi_{\le E}$ of the state $\rho$ onto the Hilbert space $\bar{\mathcal H}$ of states with less than $E$ photons fails is bounded as
\be
\ba
P_r&=\binom nr\Tr\left[\left(1-\Pi_{\le E}\right)T^r(1-T)^{n-r}\rho^{\otimes(n+1)}\right]\\
&=\binom{n}{r}\Tr\left[UT^r(1-T)^{n-r}\rho^{\otimes(n+1)}\right]\\
&\le2\binom{n}{r}\Tr\left(T\rho\right)^{r+1}\Tr\left[(1-T)\rho\right]^{n-r}\\
&\le2\binom{n}{r}\max_p{\left|p^{r+1}(1-p)^{n-r}\right|}\\
& = 2\binom{n}{r}\left(\frac{r+1}{n+1}\right)^{r+1}\left(1-\frac{r+1}{n+1}\right)^{n-r}\\
&\le\frac{2n^r}{r!}\left(\frac{r+1}{n+1}\right)^{r+1}\left(1-\frac{r+1}{n+1}\right)^{n-r}\\
&\le\frac{2n^r}{r!}\frac{(r+1)^{r+1}}{n^{r+1}}\exp\left[{-\frac{(n-r)(r+1)}{n+1}}\right]\\
&\le\frac{2}{n}\frac{r+1}{\sqrt{2\pi(r+1)}}\exp\left[{\frac{(r+1)^2}{n+1}}\right]\\
&\le\frac{\sqrt{r+1}}{n}\exp\left[{\frac{(r+1)^2}{n+1}}\right],
\ea
\ee
where we used $1-x\le e^{-x}$ and $(r+1)!\ge\sqrt{2\pi(r+1)}(r+1)^{r+1}e^{-(r+1)}$. For $s\in\mathbb N$, and for all $r\le s$,
\be
\frac{\sqrt{r+1}}{n}\exp\left[{\frac{(r+1)^2}{n+1}}\right]\le\frac{\sqrt{s+1}}{n}\exp\left[\frac{(s+1)^2}{n+1}\right],
\ee
hence the probability that at most $s$ among $n$ values of $|\alpha_i|^2$ are bigger than $E$, and that the projection $\Pi_{\le E}$ of the state $\rho$ onto the Hilbert space $\bar{\mathcal H}$ of states with less than $E$ photons fails is bounded by
\be
P^{iid}_{\text{support}}:=\frac{(s+1)^{3/2}}{n}\exp\left[{\frac{(s+1)^2}{n+1}}\right].
\label{Etesti.i.d.}
\ee
%
For $1\ll s\ll n$, this implies that either $\rho$ is contained in a lower dimensional subspace, or the score at the support estimation step is higher than $s$, with high probability.

\subsection{Proof of Theorem~\ref{thi.i.d.}}
\label{app:proofi.i.d.}

\noindent Let $\epsilon,\epsilon'>0$, let $n\ge1$, and let $\bm{\alpha}=\alpha_1,\dots,\alpha_n$ be samples obtained by measuring with heterodyne detection $n$ copies of a state $\rho$. Let $E$ in $\mathbb N$, and let $r$ be the number of samples such that $|\alpha_i|^2>E$. Let also $\ket\Psi$ be a pure state. Let us define
\be
F_\Psi(\rho)=\left[\frac{1}{n}\sum_{i=1}^n{f_\Psi\left(\alpha_i,\frac\epsilon{mK_\Psi}\right)}\right]^m,
\label{tildeFapp}
\ee
where the function $f$ and the constant $K$ are defined in Eqs.~(\ref{f}) and (\ref{K}).
Then:

\begin{theo} For all $m,s\le n$,
\be
\left|F(\Psi^{\otimes m},\rho^{\otimes m})-F_\Psi(\rho)\right|\le\epsilon+\epsilon',
\ee
or $r>s$, with probability greater than
\be
1-\left(P_{\text{support}}^{iid}+P_{\text{Hoeffding}}^{iid}\right),
\ee
\vspace{-5pt}
where
\vspace{-5pt}
\be
P_{\text{support}}^{iid}=\frac{(s+1)^{3/2}}{n}\exp\left[{\frac{(s+1)^2}{n+1}}\right],
\ee
and
\be
P_{\text{Hoeffding}}^{iid}=2\exp\left[{-\frac{n\epsilon^{2+2E}\epsilon'^2}{2m^{4+2E}C^2_\Psi}}\right],
\ee
with
\be
C_\Psi=\sum_{k,l=0}^E{|\psi_k\psi_l|\left(\frac\epsilon m\right)^{E-\frac{k+l}2}K_\Psi^{1+\frac{k+l}2}\sqrt{2^{|l-k|}\binom{\max{(k,l)}}{\min{(k,l)}}}}\underset{\epsilon\rightarrow0}{\longrightarrow}|\psi_E|^2K_\Psi^{1+E}.
\ee
\end{theo}

\noindent We first consider the case of $m=1$, from which we deduce the general case.\\

Let us write $\ket\Psi=\sum_{n\ge0}\psi_n\ket n$. For $\eta>0$, the function $z\mapsto f_\Psi(z,\eta)$ is real-valued, since $\ket\Psi\bra\Psi$ is hermitian. It is bounded as
\begin{align*}
\left|f_{\Psi}(\alpha,\eta)\right|&=\left|\sum_{k,l=0}^E{\psi_k\psi_l^*f_{\ket k\bra l}(\alpha,\eta)}\right|\\
&\le\sum_{k,l=0}^E{\left|\psi_k\psi_l^*f_{\ket k\bra l}(\alpha,\eta)\right|}\\
&\le\sum_{k,l=0}^E{\left|\psi_k\psi_l\right|\frac{M_{kl}}{\eta^{1+\frac{k+l}2}}}\\
&=\frac{1}{\eta^{1+E}}\sum_{k,l=0}^E{\left|\psi_k\psi_l\right|\eta^{E-(k+l)/2}M_{kl}}\\
&=\frac{M_\Psi(\eta)}{\eta^{1+E}},
\end{align*}
where we used Lemma~\ref{lem:boundfkl}, and where we defined
\be
M_\Psi(\eta):=\sum_{k,l=0}^E{\left|\psi_k\psi_l\right|\eta^{E-(k+l)/2}M_{kl}}.
\label{boundftau}
\ee
Applying Lemma~\ref{lem:HoeffdingR} to the real-valued function $z\mapsto f_\Psi(z,\eta)$ thus yields
\be
\Pr\left[\left|\frac{1}{n}\sum_{i=1}^n{f_{\Psi}(\alpha_i,\eta)}-\underset{\alpha\leftarrow Q_{\mathrlap\rho}}{\mathbb E}[f_{\Psi}(\alpha,\eta)]\right|\ge\epsilon'\right] \leq 2\exp\left[{-\frac{n\eta^{2+2E}\epsilon'^2}{2M_\Psi^2(\eta)}}\right],
\label{apphoeffding}
\ee
for $\epsilon',\eta>0$, where the probability is over i.i.d.\@ samples from heterodyne detection of $\rho$.

\medskip

In what follows, we first assume that $\rho\in\bar{\mathcal H}$. By Corollary~\ref{corofidelity} we have
\be
\left|F\left(\Psi,\rho\right)-\underset{\alpha\leftarrow Q_{\mathrlap\rho}}{\mathbb E}[f_{\Psi}(\alpha,\eta)]\right|\le\eta K_\Psi.
\label{appth4}
\ee
Combining Eqs.~(\ref{apphoeffding}) and (\ref{appth4}) yields
\be
\left|F\left(\Psi,\rho\right)-\frac{1}{n}\sum_{i=1}^n{f_{\Psi}(\alpha_i,\eta)}\right|\le\eta K_\Psi+\epsilon',
\ee
with probability greater than $1-2\exp\left[{-\frac{n\eta^{2+2E}\epsilon'^2}{2M_\Psi^2(\eta)}}\right]$.
Setting $\eta=\frac\epsilon{K_\Psi}$ yields
\be
\left|F\left(\Psi,\rho\right)-\frac{1}{n}\sum_{i=1}^n{f_{\Psi}\left(\alpha_i,\frac\epsilon{K_\Psi}\right)}\right|\le\epsilon+\epsilon',
\label{m1}
\ee
with probability greater than $1-2\exp\left[{-\frac{n\epsilon^{2+2E}\epsilon'^2}{2C_{\Psi,1}^2(\epsilon)}}\right]$,
where we defined
\be
\ba
C_{\Psi,1}(\epsilon)&:=K_\Psi^{1+E}M_\Psi\left(\frac\epsilon{K_\Psi}\right)\\
&=\sum_{k,l=0}^E{|\psi_k\psi_l|\epsilon^{E-\frac{k+l}2}K_\Psi^{1+\frac{k+l}2}\sqrt{2^{|l-k|}\binom{\max{(k,l)}}{\min{(k,l)}}}}.
\ea
\label{CtauE}
\ee
Combining Lemma~\ref{lem:simplem} and Eq.~(\ref{m1}) we obtain
\be
\left|F\left(\Psi,\rho\right)^m-\left[\frac{1}{n}\sum_{i=1}^n{f_{\Psi}\left(\alpha_i,\frac\epsilon{K_\Psi}\right)}\right]^m\right|\le m(\epsilon+\epsilon'),
\ee
with probability greater than $1-2\exp\left[{-\frac{n\epsilon^{2+2E}\epsilon'^2}{2C^2_{\Psi,1}(\epsilon)}}\right]$.
Note that we excluded the pathological case $\frac{1}{n}\sum_{i=1}^n{f_{\Psi}(\alpha_i,\epsilon/ K_\Psi)}>1$:
when that is the case we instead set \mbox{$\frac{1}{n}\sum_{i=1}^n{f_{\Psi}(\alpha_i,\epsilon/ K_\Psi)}=1$}. 
The target state $\Psi$ is pure so $F(\Psi^{\otimes m},\rho^{\otimes m})=F\left(\Psi,\rho\right)^m$. Hence, replacing $\epsilon$ and $\epsilon'$ by $\epsilon/m$ and $\epsilon'/m$, respectively, gives
\be
\left|F\left(\Psi,\rho\right)^m-F_\Psi(\rho)\right|\le\epsilon+\epsilon',
\ee
with probability greater than
\be
P_{\text{Hoeffding}}^{iid}:=1-2\exp\left[{-\frac{n\epsilon^{2+2E}\epsilon'^2}{2m^{4+2E}C^2_\Psi}}\right].
\ee
where $F_\Psi(\rho)=\left[\frac{1}{n}\sum_{i=1}^n{f_\Psi\left(\alpha_i,\frac\epsilon{mK_\Psi}\right)}\right]^m$ and where
\be
\ba
C_\Psi&:=C_{\Psi,1}(\epsilon/m)\\
&=\sum_{k,l=0}^E{|\psi_k\psi_l|\left(\frac\epsilon m\right)^{E-\frac{k+l}2}K_\Psi^{1+\frac{k+l}2}\sqrt{2^{|l-k|}\binom{\max{(k,l)}}{\min{(k,l)}}}}.
\ea
\label{Cpsiproof}
\ee
Until now we have assumed $\rho\in\bar{\mathcal H}$. By Eq.~(\ref{Etesti.i.d.}), the probability that at most $s$ among $n$ values of $|\alpha_i|^2$ are bigger than $E$, and that the projection $\Pi_{\le E}$ of the state $\rho$ onto the Hilbert space $\bar{\mathcal H}$ of states with less than $E$ photons fails is bounded by
\be
P^{iid}_{\text{support}}=\frac{(s+1)^{3/2}}{n}\exp\left[{\frac{(s+1)^2}{n+1}}\right].
\ee
With the union bound we thus obtain
\be
\left|F\left(\Psi,\rho\right)^m-F_\Psi(\rho)\right|\le\epsilon+\epsilon',
\ee
or $r>s$, with probability greater than $1-\left(P^{iid}_{\text{support}}+P_{\text{Hoeffding}}^{iid}\right)$.

\qed


\section{Proof of Theorem~\ref{thVUCVQC}}
\label{app:verification}

\noindent In this section, we detail the proof of Theorem~\ref{thVUCVQC}. Compared to the case of the previous section, the i.i.d.\@ states are replaced with permutation-invariant states, which makes the proof more technical. We recall the protocol and the theorem in the following.\\

The verifier wants to verify $m$ copies of a target pure state $\ket\Psi$. The numbers $n$, $k$, $q$, and $E$ are free parameters of the protocol. The verifier instructs the prover to prepare $n+k$ copies of $\ket\Psi$. Let us write $\rho^{n+k}$ the state received by the verifier. It picks $k$ subsystems at random and measures them with heterodyne detection, obtaining the samples $\beta_1,\dots,\beta_k$. It records the number $r$ of values $|\beta_i|^2>E$. It discards $4q$ subsystems at random and measures all the others but $m$ chosen at random with heterodyne detection, obtaining the samples $\alpha_1,\dots,\alpha_{n-4q-m}$. Finally, the verifier computes with these samples the estimate
\be
F_\Psi(\rho)=\left[\frac{1}{n-4q-m}\sum_{i=1}^{n-4q-m}{f_\Psi\left(\alpha_i,\frac\epsilon{mK_\Psi}\right)}\right]^m,
\label{tildeF2app}
\ee
where the function $f$ and the constant $K$ are defined in Eqs.~(\ref{f}) and (\ref{K}), and where $\epsilon>0$ is a free parameter.
 
\begin{theo} Let $\epsilon,\epsilon'>0$. For all $s\le k$,
\be
\left|F\left(\Psi^{\otimes m},\rho^m\right)-F_\Psi(\rho)\right|\le\epsilon+\epsilon'+P_{\text{deFinetti}},
\ee
or $r>s$, with probability greater than
\be
1-\left(P_{\text{support}}+P_{\text{deFinetti}}+P_{\text{choice}}+P_{\text{Hoeffding}}\right),
\ee
where
\be
P_{\text{support}}=8k^{3/2}\exp\left[{-\frac{k}{9}\left(\frac{q}{n}-\frac{2s}{k}\right)^2}\right],
\ee
\be
P_{\text{deFinetti}}=q^{(E+1)^2/2}\exp\left[{-\frac{2q(q+1)}{n}}\right],
\ee
\be
P_{\text{choice}}=\frac{m(4q+m-1)}{n-4q},
\ee
and
\be
P_{\text{Hoeffding}}=2\binom{n-4q}{4q}\exp\left[{-\frac{n-8q}{2m^{4+2E}}\left(\frac{\epsilon^{1+E}\epsilon'}{C_\Psi}-\frac{8qm^{2+E}}{n-4q-m}\right)^2}\right],
\ee
with
\be
C_\Psi=\sum_{k,l=0}^E{|\psi_k\psi_l|\left(\frac\epsilon m\right)^{E-\frac{k+l}2}K_\Psi^{1+\frac{k+l}2}\sqrt{2^{|l-k|}\binom{\max{(k,l)}}{\min{(k,l)}}}}\underset{\epsilon\rightarrow0}{\longrightarrow}|\psi_E|^2K_\Psi^{1+E}.
\ee
\end{theo}

\noindent The proof builds upon and generalises results from~\cite{renner2008security,renner2008finetti,renner2009finetti}, and follows these steps:

\begin{itemize}
\item \textit{Support estimation}: with probability arbitrarily close to $1$, most of the subsystems of the permutation-invariant state $\rho^{n-q}$ lie in a lower dimensional subspace, or the score of the state $\rho^{n+k}$ at the support estimation step is high (section~\ref{app:Etestsym}) ;
\item \textit{De Finetti reduction}: any permutation-invariant state with most of its subsystems in a lower dimensional subspace admits a purification in the symmetric subspace that still has most of its subsystems in a lower dimensional subspace. This purification is well approximated by a mixture of almost-i.i.d.\@ states (section~\ref{app:deFinetti}) ;
\item \textit{Hoeffding inequality for almost-i.i.d.\@ states}: mixture of almost-i.i.d.\@ states can be certified in a similar fashion as i.i.d.\@ states (section~\ref{app:almost-i.i.d.}).
\end{itemize}

\noindent We conclude the proof in section~\ref{app:finalproof} by combining the previous points.

\subsection{Support estimation for permutation-invariant states}
\label{app:Etestsym}

\noindent We first derive a support estimation step for permutation-invariant states.
We will use in this section the following operators, already introduced in Section~\ref{app:Etesti.i.d.}: for $E\ge0$:
\be
U=\smashoperator{\sum_{n=E+1}^{+\infty}}{\ket n\bra n}=1-\Pi_{\le E},
\ee
where $\Pi_{\le E}=\smashoperator{\sum_{n=0}^E}{\ket n\bra n}$ is the projector onto 
the Hilbert space $\bar{\mathcal H}$ of states with at most $E$ photons, and
\be
T=\frac{1}{\pi}\smashoperator{\int_{\quad|\alpha|^2\ge E}}{\ket\alpha\bra\alpha d^2\alpha},
\ee
where $\ket\alpha$ is a coherent state. We also recall the following result, from Eq.~(\ref{U2T}), proven in~\cite{leverrier2013security}:
\be
U\le2T.
\label{UleT}
\ee
We recall a few notations and results from~\cite{renner2008finetti}: let $\mathcal A=\{A_0,A_1\},\mathcal B=\{B_0,B_1\}$ be two binary POVMs over $\mathcal H$. Define for $\delta>0$,
\be
\gamma_{A\rightarrow B}(\delta)=\underset{\Psi}{\sup}{\left\{\Tr(B\Psi),\text{s.t.}\Tr(A\Psi)\le\delta\right\}}.
\ee
In particular,
\be
\gamma_{T\rightarrow U}(\delta)\le2\delta,
\label{gammaTU}
\ee
by Eq.~(\ref{UleT}). We recall the following result (Lemma~III.1. of~\cite{renner2008finetti}):

\begin{lem} 
Let $n\ge2k$, let $\delta>0$, let $\mathcal A=\{A_0,A_1\}$ and $\mathcal B=\{B_0,B_1\}$ be two binary POVMs over $\mathcal H$, and let $x_1,\dots,x_{n+k}$ the $(n+k)$-partite classical outcome of the measurement $\mathcal{A}^{\otimes n}\otimes \mathcal{B}^{\otimes k}$ applied to any permutation-invariant state $\rho^{n+k}$. Then
\be
\Pr\left[\frac{x_{1}+\dots+x_{n}}{n}>\gamma_{B_1\rightarrow A_1}\left(\frac{x_{n+1}+\dots+x_{n+k}}{k}+\delta\right)+\delta\right]\le8k^{3/2}e^{-k\delta^2}.
\ee
\label{lem:Serf2}
\end{lem}

\noindent This result is a refined version of Serfling's bound~\cite{serfling1974probability}. It relates the outcomes of a measurement on some subsystems of a symmetric state with the outcomes of a related measurement on the rest of the subsystems. With this technical Lemma, we derive in what follows a support estimation step for permutation-invariant states using samples from heterodyne detection.\\

Let $\rho^{n+k}$ be a state over $n+k$ subsystems. Applying a random permutation to this state and measuring its last $k$ subsystems with heterodyne detection is equivalent to measuring $k$ subsystems at random. We thus assume in the following that the state $\rho^{n+k}$ is a permutation-invariant state, without loss of generality, and that the verifier measures its last $k$ subsystems with heterodyne detection. \\

Let $\mathcal T=\{1-T,T\}$ and $\mathcal U=\{1-U,U\}$. Let $x_1,\dots,x_{n+k}$ the $(n+k)$-partite classical outcome of the measurement $\mathcal{U}^{\otimes n}\otimes \mathcal{T}^{\otimes k}$ applied to the permutation-invariant state $\rho^{n+k}$ sent by the prover. A value $x_i=1$ for $i\in1,\dots,n$ means that the projection of the $i^{th}$ subsystem onto $\bar{\mathcal H}$ failed, while a value $x_j=1$ for $j\in n+1,\dots,n+k$ means that the value $|\beta|^2$ obtained when measuring the $j^{th}$ subsystem with heterodyne detection was bigger than $E$. In particular, the number of values $\beta_i$ satisfying $|\beta_i|^2>E$, is expressed as $x_{n+1}+\dots+x_{n+k}$.
Let $\mathcal T_{\le s}^k$ be the event that at most $s$ of the $k$ values $\beta_i$ satisfy $|\beta_i|^2>E$, and let $\mathcal F_q^n$ be the event that the projection onto $\bar{\mathcal H}$ fails for more than $q$ subsystems of the remaining state $\rho^n$. Then:

\begin{lem} 
\be
\Pr\left[\mathcal F_q^n\cap\mathcal T_{\le s}^k\right]\le P_{\text{support}}.
\ee
where $P_{\text{support}}=8k^{3/2}\exp\left[{-\frac{k}9\left(\frac qn-\frac{2s}k\right)^2}\right]$.
\label{lem:supportpi}
\end{lem}

\begin{leftbar}
\begin{proof} With Eq.~(\ref{gammaTU}), we have for all $\delta>0$
\be
\gamma_{T\rightarrow U}\left(\frac{x_{n+1}+\dots+x_{n+k}}{k}+\delta\right)+\delta\le2\frac{x_{n+1}+\dots+x_{n+k}}{k}+3\delta.
\ee
Taking $\delta_0=\frac{1}{3}\left(\frac qn-\frac{2s}k\right)$ we obtain
\be
\gamma_{T\rightarrow U}\left(\frac{x_{n+1}+\dots+x_{n+k}}{k}+\delta_0\right)+\delta_0\le\frac qn+2\left(\frac{x_{n+1}+\dots+x_{n+k}}{k}-\frac sk\right),
\ee
so if $x_{1}+\dots+x_{n}>q$ and $x_{n+1}+\dots+x_{n+k}\le s$, then
\be
\gamma_{T\rightarrow U}\left(\frac{x_{n+1}+\dots+x_{n+k}}{k}+\delta_0\right)+\delta_0<\frac{x_{1}+\dots+x_{n}}{n}.
\ee
Hence,
\be
\ba
\Pr\left[\mathcal F_q^n\cap\mathcal T_{\le s}^k\right]&=\Pr\left[\left(x_{1}+\dots+x_{n}>q\right)\cap\left(x_{n+1}+\dots+x_{n+k}\le s\right)\right]\\
&\le\Pr\left[\left(\frac{x_{1}+\dots+x_{n}}{n}>\gamma_{T\rightarrow U}\left(\frac{x_{n+1}+\dots+x_{n+k}}{k}+\delta_0\right)+\delta_0\right)\right]\\
&\le8k^{3/2}e^{-k\delta_0^2}\\
&=8k^{3/2}\exp\left[{-\frac{k}9\left(\frac qn-\frac{2s}k\right)^2}\right],
\ea
\ee
where we used Lemma~\ref{lem:Serf2} for $\mathcal A=\mathcal U$ and $\mathcal B=\mathcal T$.

\end{proof}
\end{leftbar}

\noindent Recall that $\bar{\mathcal H}$ is the Hilbert space of states with at most $E$ photons, of dimension $E+1$.
For $q\le n$, let us define the set of permutation-invariant states over $n$ subsystems, with at most $q$ subsystems out of this lower dimensional subspace (introduced in~\cite{renner2009finetti}):
\be
\mathcal S^n_{\bar{\mathcal H}^{\otimes n-q}}:=\text{span }\underset\pi\bigcup\text{ }\pi\left(\bar{\mathcal H}^{\otimes n-q}\otimes\mathcal H^{\otimes q}\right)\pi^{-1},
\ee
where the union is taken over all permutations. Lemma~\ref{lem:supportpi} then gives
\be
\Pr\left[\mathcal F_q^n\cap\mathcal T_{\le s}^k\right]\le P_{\text{support}},
\ee
where $\mathcal F_q^n$ is the event that the projection of $\rho^n$ (the remaining state after the support estimation step) onto $\mathcal S^n_{\bar{\mathcal H}^{\otimes n-q}}$ fails, and where $P_{\text{support}}=8k^{3/2}\exp\left[{-\frac{k}9\left(\frac qn-\frac{2s}k\right)^2}\right]$. For $1\ll q\ll n$ and $q/n\ll s/k$, this implies that either $\rho^n$ has most of its subsystems in a lower dimensional subspace, or the score at the support estimation step is higher than $s$, with high probability.

\subsection{De Finetti reduction}
\label{app:deFinetti}

\noindent We recall in this section two results from~\cite{renner2009finetti}.

\begin{itemize}

\item The first result says that any permutation-invariant state with most of its subsystems in a lower dimensional subspace has a purification in the symmetric subspace that still has most of its subsystems in a lower dimensional subspace. Formally, for $n\in\mathbb N$, and given a Hilbert space $\mathcal K$, let us write Sym$^n(\mathcal K)=\{\phi\in\mathcal K^{\otimes n},\text{ }\pi\phi=\phi\text{ }(\forall \pi)\}$ the symmetric subspace of a Hilbert space $\mathcal K^{\otimes n}$, then (Lemma 3 of~\cite{renner2009finetti}):

\begin{lem} For all $q\le n$, any permutation-invariant state $\rho^n\in\mathcal S^n_{\bar{\mathcal H}^{\otimes n-q}}$ has a purification $\tilde\rho^n$ in \mbox{$\mathrm{Sym}^n(\mathcal H\otimes\mathcal H)\bigcap\mathcal S^n_{(\bar{\mathcal H}\otimes\bar{\mathcal H})^{\otimes n-2q}}$}.
\label{lem:purif}
\end{lem}

\end{itemize}

\noindent The states of the form $\ket v^{\otimes n}$ are the so-called \textit{i.i.d.\@ states}. For all $n,r\ge0$ and all $\ket v\in\bar{\mathcal H}\otimes\bar{\mathcal H}$, the set of \textit{almost-i.i.d.\@ states along $\ket v$}, $\mathcal S^n_{v^{\otimes n-r}}$, is defined as the span of all vectors that are, up to reorderings, of the form $\ket v^{\otimes n-r}\otimes\ket\phi$, for an arbitrary $\phi\in\left(\mathcal H\otimes\mathcal H\right)^{\otimes r}$. In the following, we simply refer to these states as \textit{almost-i.i.d.\@ states} (which becomes relevant when $r\ll n$).

\begin{itemize}

\item The second result is a de Finetti theorem for states in $\mathrm{Sym}^n(\mathcal H\otimes\mathcal H)\bigcap\mathcal S^n_{(\bar{\mathcal H}\otimes\bar{\mathcal H})^{\otimes n-2q}}$, which says that reduced states from them are well approximated by mixtures of almost-i.i.d.\@ states. Formally (Theorem 4 of~\cite{renner2009finetti}, applied to $\mathcal K=\mathcal H\otimes\mathcal H$ and $\bar{\mathcal K}=\bar{\mathcal H}\otimes\bar{\mathcal H}$, with dim$(\bar{\mathcal K})=(E+1)^2$):

\begin{theo}
Let $\tilde\rho^n\in\mathrm{Sym}^n(\mathcal H\otimes\mathcal H)\bigcap\mathcal S^n_{(\bar{\mathcal H}\otimes\bar{\mathcal H})^{\otimes n-2q}}$ and let $\tilde\rho^{n-4q}=\mathrm{Tr}_{4q}(\tilde\rho^n)$. Then, there exist a finite set $\mathcal V$ of unit vectors $\ket v\in\bar{\mathcal H}\otimes\bar{\mathcal H}$, a probability distribution $\{p_v\}_{v\in\mathcal V}$ over $\mathcal V$, and almost-i.i.d.\@ states $\tilde\rho_v^{n-4q}\in\mathcal S^{n-4q}_{v^{\otimes n-8q}}$ such that
\be
F\left(\tilde\rho^{n-4q},\sum_{v\in\mathcal V}{p_v\tilde\rho_v^{n-4q}}\right)>1-q^{(E+1)^2}\exp\left[{-\frac{4q(q+1)}{n}}\right].
\ee
\label{thDeFinetti}
\end{theo}

\end{itemize}

\noindent Given a state $\rho^n\in\mathcal S^n_{\bar{\mathcal H}^{\otimes n-q}}$, applying Theorem~\ref{thDeFinetti} to the purification $\tilde\rho^n$ given by Lemma~\ref{lem:purif} shows that the reduced state $\tilde\rho^{n-4q}$ is close in fidelity to a mixture of states that are i.i.d.\@ on $n-8q$ subsystems.

\subsection{Hoeffding inequality for almost-i.i.d.\@ states}
\label{app:almost-i.i.d.}

\noindent We recall here Lemma~\ref{lem:HoeffdingR}, in the context of a product measurement applied to an i.i.d.\@ state $\ket v\bra v^{\otimes n}$:

\begin{lem}\textbf{(Hoeffding)} 
Let $M>0\in\mathbb R$ and let $f:\mathbb C\mapsto\mathbb R$ be a function bounded as $|f(\alpha)|<M$ for all $\alpha\in\mathbb C$. Let $\lambda>0$, let $p\in\mathbb N^*$, and let $\ket v\in \mathcal H$. Let $\mathcal M=\{\mathcal M_\alpha\}_{\alpha\in\mathbb C}$ be a POVM on $\mathcal H$ and let $D_{\ket v}$ be the probability density function of the outcomes of the measurement $\mathcal M$ applied to $\ket v\bra v$. Then
\begin{equation}
\underset{\bm{\alpha}}{\Prb}\left[\left|\frac{1}{p}\sum_{i=1}^p{f(\alpha_i)}- \underset{\beta\leftarrow D_{\mathrlap{\ket v}}}{\mathbb E}[f(\beta)]\right|\ge\lambda\right] 
\leq 2\exp\left[{-\frac{p\lambda^2}{2M^2}}\right],
\end{equation}
where the probability is taken over the outcomes $\bm{\alpha}=(\alpha_1,\dots,\alpha_p)$ of the product measurement $\mathcal M^{\otimes p}$ applied to $\ket v\bra v^{\otimes p}$.
\label{lem:Hoeffding2}
\end{lem}

\noindent The next result gives an equivalent statement for almost-i.i.d.\@ states along a state $\ket v$, measured with a product measurement. It generalises Theorem 4.5.2 of~\cite{renner2008security}, where the probability distributions over finite sets, corresponding to product measurements with finite number of outcomes, are replaced by continuous variable probability densities, corresponding to product measurements with continuous variable outcomes. Frequencies estimators are also replaced with estimators of expected values of bounded functions. We will use this result for the POVM corresponding to a product heterodyne detection.

\begin{lem} 
Let $M>0\in\mathbb R$ and let $f:\mathbb C\mapsto\mathbb R$ be a function bounded as $|f(\alpha)|\le M$ for all $\alpha\in\mathbb C$. Let $\mu>0$ and $1\le m\le r<t$ such that
\be
(t-m)\mu>2Mr.
\ee
Let also $\ket v\in\bar{\mathcal H}$ and $\ket\Phi\in\mathcal S^t_{v^{\otimes t-r}}$. Let $\mathcal M=\{\mathcal M_\alpha\}_{\alpha\in\mathbb C}$ be a POVM on $\mathcal H$ and let $D_{\ket v}$ be the probability density function of the outcomes of the measurement $\mathcal M$ applied to $\ket v\bra v$. Then
\begin{equation}
\underset{\bm{\alpha}}{\Prb}\left[\left|\frac{1}{t-m}\sum_{i=1}^{t-m}{f(\alpha_i)}- \underset{\beta\leftarrow D_{\ket v}}{\mathbb E}[f(\beta)]\right|\ge\mu\right] \leq 2\binom{t}{r}\exp\left[{-\frac{t-r}2\left(\frac{\mu}{M}-\frac{2r}{t-m}\right)^2}\right],
\end{equation}
where the probability is taken over the outcomes $\bm{\alpha}=(\alpha_1,\dots,\alpha_{t-m})$ of the product measurement $\mathcal M^{\otimes t-m}$ applied to $\ket\Phi\bra\Phi$.
\label{lem:almosti.i.d.}
\end{lem}

\noindent In essence, this lemma says that a product measurement on all but $m$ subsystems of an almost-i.i.d.\@ state along a state $\ket v$ will yield statistics that are similar to the ones that would be obtained by measuring the i.i.d.\@ state $\ket v^{\otimes t-m}$.

\begin{leftbar}
\begin{proof} 
$\ket\Phi\in\mathcal S^t_{v^{\otimes t-r}}$, so by Lemma 4.1.6 of~\cite{renner2008security}, there exist a finite set $\mathcal S$ of size at most $\binom{t}{r}$, a family of states $\ket{\tilde\Phi^s}\in\mathcal H^{\otimes r}$ for $s\in\mathcal S$, complex amplitudes $\{\gamma_s\}_{s\in\mathcal S}$ and permutations $\{\pi_s\}_{s\in\mathcal S}$ over $[1,\dots,t]$ such that
\be
\ba
\ket\Phi&:=\sum_{s\in\mathcal S}{\gamma_s\ket{\Phi^s}}\\
&=\sum_{s\in\mathcal S}{\gamma_s\pi_s\left(\ket v^{\otimes t-r}\otimes\ket{\tilde\Phi^s}\right)}.
\ea
\label{Psidec}
\ee
With the notations of the Lemma, let us define for $\mu>0$:
\be
\Omega_\mu=\left\{\bm{\alpha}\in\mathbb C^{t-m},\left|\frac{1}{t-m}\sum_{i=1}^{t-m}{f(\alpha_i)}- \underset{\beta\leftarrow D_{\ket v}}{\mathbb E}[f(\beta)]\right|>\mu\right\}.
\ee
We recall here Lemma of 4.5.1 of~\cite{renner2008security}:

\begin{lem}
Let $|\mathcal X|$ be a finite set and $\ket\psi=\sum_{x\in\mathcal X}{\ket{\psi^x}}$, and let $A$ be a non-negative operator. Then
\be
\braket{\psi|A|\psi}\le|\mathcal X|\sum_{x\in\mathcal X}{\braket{\psi^x|A|\psi^x}}.
\ee
\end{lem}

\noindent In particular, using Eq.~(\ref{Psidec}) and this Lemma when $A$ is a POVM element of the product measurement \mbox{$\mathcal M_{\bm\alpha}\equiv\mathcal M_{\alpha_1}\otimes\dots\otimes\mathcal M_{\alpha_{t-m}}$}, we obtain:
\be
\ba
\underset{\bm{\alpha}\leftarrow\ket\Phi}{\Prb}[\bm{\alpha}\in\Omega_\mu]&=\int_{\mathrlap{\Omega_\mu}}{\braket{\Phi|\mathcal M_{\bm{\alpha}}|\Phi}d^{2(t-m)}\bm{\alpha}}\\
&\le\int_{\mathrlap{\Omega_\mu}}{|\mathcal S|\sum_{s\in\mathcal S}{|\gamma_s|^2\braket{\Phi^s|\mathcal M_{\bm{\alpha}}|\Phi^s}}d^{2(t-m)}\bm{\alpha}}\\
&\le|\mathcal S|\sum_{s\in\mathcal S}{|\gamma_s|^2\int_{\mathrlap{\Omega_\mu}}{\braket{\Phi^s|\mathcal M_{\bm{\alpha}}|\Phi^s}}d^{2(t-m)}\bm{\alpha}}\\
&=|\mathcal S|\sum_{s\in\mathcal S}{|\gamma_s|^2\underset{\bm{\alpha}\leftarrow\ket{\Phi^s}}{\Prb}[\bm{\alpha}\in\Omega_\mu]},
\ea
\label{PrOmega}
\ee
where we write $\bm{\alpha}\leftarrow\ket\chi$  to indicate that $\bm{\alpha}=(\alpha_1,\dots,\alpha_{t-m})$ is distributed according to the outcomes of the product measurement $\mathcal M^{\otimes(t-m)}$ applied to $\ket\chi$.\\ \\

Let $\bm{\alpha}\leftarrow\ket{\Phi^s}$. We have $\ket{\Phi^s}=\pi_s(\ket v^{\otimes t-r}\otimes\ket{\tilde\Phi^s})$, and in particular $(\alpha_{\pi_s(1)},\dots,\alpha_{\pi_s(t-r)})$ is distributed according to the outcomes of the product measurement $\mathcal M^{\otimes t-r}$ applied to $\ket{v}^{\otimes t-r}$. 
We also have, for $|f|\le M$,
\be
\ba
\quad&\left|\frac{1}{t-r}\sum_{i=1}^{t-r}{f(\alpha_{\pi_s(i)})}-\frac{1}{t-m}\sum_{i=1}^{t-m}{f(\alpha_i)}\right|\\
&=\left|\frac{1}{t-r}\sum_{i=1}^{t-r}{f(\alpha_{\pi_s(i)})}
  -\frac{1}{t-m}\left(\sum_{i=1}^{t}{f(\alpha_i)}
                     -\smashoperator{\sum_{i=t-m+1}^{t}}{f(\alpha_i)}\right)\right|\\
&=\left|\frac{1}{t-r}\sum_{i=1}^{t-r}{f(\alpha_{\pi_s(i)})}
  -\frac{1}{t-m}\left(\sum_{i=1}^{t}{f(\alpha_{\pi_s(i)})}
                     -\smashoperator{\sum_{i=t-m+1}^{t}}{f(\alpha_i)}\right)\right|\\
&=\left|\left(\frac{1}{t-r}-\frac{1}{t-m}\right)\sum_{i=1}^{t-r}{f(\alpha_{\pi_s(i)})}
  +\frac{1}{t-m}\left(\sum_{i=t\mathrlap{-m+1}}^{t}{f(\alpha_i)}
                     -\sum_{i=t\mathrlap{-r+1}}^{t}{f(\alpha_{\pi_s(i)})}\right)\right|\\
&\le\left|\frac{1}{t-r}-\frac{1}{t-m}\right|\sum_{i=1}^{t-r}{|f(\alpha_{\pi_s(i)})|}
  +\frac{1}{t-m}\left(\sum_{i=t\mathrlap{-m+1}}^{t}{|f(\alpha_i)|}
                     +\sum_{i=t\mathrlap{-r+1}}^{t}{|f(\alpha_{\pi_s(i)})|}\right)\\
&\le\frac{|r-m|}{t-m}M+\frac{(m+r)}{t-m}M\\
&=\frac{2rM}{t-m},
\ea
\label{boundfreq}
\ee
where we used $r\ge m$.\newpage
Now for all $s\in\mathcal S$,  
\be
\ba
\underset{\bm{\alpha}\leftarrow\ket{\Phi^s}}{\Prb}[\bm{\alpha}\in\Omega_\mu]
 &=\underset{\bm{\alpha}\leftarrow\ket{\Phi^s}}{\Prb}\left[
    \left|\frac{1}{t-m}\sum_{i=1}^{t-m}{f(\alpha_i)}
   -\underset{\beta\leftarrow D_{\mathrlap{\ket v}}}{\mathbb E}[f(\beta)]\right|>\mu\right]\\
&\le\underset{\bm{\alpha}\leftarrow\ket{\Phi^s}}{\Prb}\Bigg[
    \left|\frac{1}{t-r}\sum_{i=1}^{t-r}{f(\alpha_{\pi_s(i)})}
    -\underset{\beta\leftarrow D_{\mathrlap{\ket v}}}{\mathbb E}[f(\beta)]\right|\\
   &\quad+\left|\frac{1}{t-m}\smashoperator{\sum_{i=1}^{t-m}}{f(\alpha_i)}
    -\frac{1}{t-r}\sum_{i=1}^{t-r}{f(\alpha_{\pi_s(i)})}\right|>\mu\Bigg]\\
&\le\underset{\bm{\alpha}\leftarrow\ket{\Phi^s}}{\Prb}\left[
    \left|\frac{1}{t-r}\sum_{i=1}^{t-r}{f(\alpha_{\pi_s(i)})}
     -\underset{\beta\leftarrow D_{\mathrlap{\ket v}}}{\mathbb E}[f(\beta)]\right|
    >\mu-\frac{2rM}{t-m}\right]\\
&\le2\exp\left[{-\frac{t-r}2\left(\frac{\mu}{M}-\frac{2r}{t-m}\right)^2}\right],
\ea
\ee
where we used triangular inequality in the second line, Eq.~(\ref{boundfreq}) in the third line and Lemma~\ref{lem:Hoeffding2} in the fourth line with $p=t-r$ and $\lambda=\mu-\frac{2rM}{t-m}>0$. Combining this last equation with Eq.~(\ref{PrOmega}), and using $|\mathcal S|\le\binom{t}{r}$ we finally obtain,
\be
\underset{\bm{\alpha}\leftarrow\ket\Psi}{\Prb}\left[\left|\frac{1}{t-m}\sum_{i=1}^{t-m}{f(\alpha_i)}- \underset{\beta\leftarrow D_{\ket v}}{\mathbb E}[f(\beta)]\right|\ge\mu\right] \leq 2\binom{t}{r}\exp\left[{-\frac{t-r}2\left(\frac{\mu}{M}-\frac{2r}{t-m}\right)^2}\right].
\ee

\end{proof}
\end{leftbar}

\noindent We recall the bound on $z\mapsto f_\Psi(z,\eta)$ obtained in Section~\ref{app:certification} for $\eta>0$, detailed in Eq.~(\ref{boundftau}): for all $\alpha\in\mathbb C$,
\be
|f_\Psi(\alpha,\eta)|\le\frac{M_\Psi(\eta)}{\eta^{1+E}},
\label{boundfM}
\ee
where
\be
M_\Psi(\eta)=\sum_{k,l=0}^E{\left|\psi_k\psi_l\right|\eta^{E-(k+l)/2}\sqrt{2^{|l-k|}\binom{\max{(k,l)}}{\min{(k,l)}}}}.
\ee
Let $\mu,\eta>0$, $E\in\mathbb N$, let $\ket v\in\bar{\mathcal H}\otimes\bar{\mathcal H}$, and let $\ket{\Phi_v}^{n-4q}\in\mathcal S^{n-4q}_{v^{\otimes n-8q}}$. Applying Lemma~\ref{lem:almosti.i.d.} for the real-valued function $f_\Psi$, for $t=n-4q$, for $r=4q$, for $D_{\ket v}=Q_{\ket v\bra v}$, and with the bound from Eq.~(\ref{boundfM}), we obtain 
\be
\ba
\quad&\underset{\bm{\alpha}}{\Prb}\left[
  \left|\frac{1}{n-4q-m}\sum_{i=1}^{n-4q-m}{f_\Psi(\alpha_i,\eta)}
  \mathrlap{-}\underset{\quad\beta\leftarrow Q_{\ket v\mathrlap{\bra v}}}{\mathbb E}[f_\Psi(\beta,\eta)]
  \right|
 \ge\mu\right] \\
&\quad\quad\quad\leq2\binom{n-4q}{4q}\exp\left[{-\frac{n-8q}2\left(\frac{\eta^{1+E}\mu}{M_\Psi(\eta)}-\frac{8q}{n-4q-m}\right)^2}\right],
\ea
\label{Hoeffdingpsiv}
\ee
where the probability is over the outcomes of a product heterodyne measurement of the first $n-4q-m$ subsystems of $\ket{\Phi_v}^{n-4q}\in\mathcal S^{n-4q}_{v^{\otimes n-8q}}$.\\

\subsection{Proof of Theorem~\ref{thVUCVQC}}
\label{app:finalproof}

\noindent Let $\ket\Psi\bra\Psi$ be the target pure state, and let $\rho^{n+k}$ be a state sent over $n+k$ subsystems. Let $\beta_1,\dots,\beta_k$ be samples obtained by measuring $k$ subsystems at random of $\rho^{n+k}$ with heterodyne detection. Let $\rho^n$ be the remaining state after the support estimation step. 
In what follows, we first assume that $\rho^n\in\mathcal S^n_{\bar{\mathcal H}^{\otimes n-q}}$.

Let $\rho^{n-4q}$ be the state obtained from $\rho^n$ by tracing over the first $4q$ subsystems. In that case, by section~\ref{app:deFinetti}, there exist a finite set $\mathcal V$ of unit vectors $\ket v\in\bar{\mathcal H}\otimes\bar{\mathcal H}$, a probability distribution $\{p_v\}_{v\in\mathcal V}$ over $\mathcal V$, and almost-i.i.d.\@ states $\tilde\rho_v^{n-4q}\in\mathcal S^{n-4q}_{v^{\otimes n-8q}}$ such that
\be
F\left(\rho^{n-4q},\sum_{v\in\mathcal V}{p_v\rho_v^{n-4q}}\right)>1-q^{(E+1)^2}\exp\left[{-\frac{4q(q+1)}{n}}\right],
\label{DeFinettifide}
\ee
where $\rho_v^{n-4q}$ is the remaining state after tracing over the purifying subsystems, since the fidelity is non-decreasing under quantum operations~\cite{barnum1996noncommuting}. We also obtain 
\be
F\left(\rho^m,\sum_{v\in\mathcal V}{p_v\rho_v^m}\right)>1-q^{(E+1)^2}\exp\left[{-\frac{4q(q+1)}{n}}\right],
\label{DeFinettifidereduced}
\ee
where $\rho^m$ (resp.\@ $\rho_v^m$) is the remaining state after measuring the first $n-4q-m$ subsystems of $\rho^{n-4q}$ (resp.\@ $\rho_v^{n-4q}$) with heterodyne detection.

Let $\alpha_1,\dots,\alpha_{n-4q-m}$ be the samples obtained by measuring the first $n-4q-m$ subsystems of $\rho^{n-4q}$ with heterodyne detection. 
The verifier computes the estimate~(\ref{tildeF2app})
\be 
F_\Psi(\rho)=\left[\frac{1}{n-4q-m}\sum_{i=1}^{n-4q-m}{f_\Psi\left(\alpha_i,\frac\epsilon{mK_\Psi}\right)}\right]^m,
\label{tildeF3}
\ee
and whenever $F_\Psi\ge1$ we instead set $F_\Psi=1$. Let us define the completely positive map $\mathcal E$ on $\mathcal H^{n-4q}$ associated to the classical post-processing of the protocol as:\\
\be
\sigma\mapsto\mathcal E(\sigma)=\sum_e{\Pr\left[F_{\Psi}(\sigma)=e\right]\ket e\bra e}.
\label{postprocessingmap}
\ee
The sum ranges over the values that the estimate may take. 
With Eq.~(\ref{DeFinettifide}) and Lemma~\ref{lem:fidelitytriangular} we obtain
\be
D\left(\rho^{n-4q},\sum_{v\in\mathcal V}{p_v\rho_v^{n-4q}}\right)\le q^{\frac{(E+1)^2}2}\exp\left[{-\frac{2q(q+1)}{n}}\right],
\label{Drhomrhov}
\ee
The trace distance is non-increasing under quantum operations, so Eq.~(\ref{Drhomrhov}) implies
\be
D\left(\mathcal E\left(\rho^{n-4q}\right),\mathcal E\left(\sum_{v\in\mathcal V}{p_v\rho_v^{n-4q}}\right)\right)\le q^{\frac{(E+1)^2}2}\exp\left[{-\frac{2q(q+1)}{n}}\right].
\ee
Using the definition of the map $\mathcal E$, we obtain a bound in total variation distance:
\be
\left\|P\left[F_\Psi(\rho)\right]-P\left[F_\Psi\left(\sum_{v\in\mathcal V}{p_v\rho_v^{n-4q}}\right)\right]\right\|_{tvd}\le q^{\frac{(E+1)^2}2}\exp\left[{-\frac{2q(q+1)}{n}}\right],
\ee
where $P$ denotes the probability distributions for the values of the estimates $F_\Psi(\rho)$ and $F_\Psi\left(\sum_{v\in\mathcal V}{p_v\rho_v^{n-4q}}\right)$. In particular, this bound implies that for all $\lambda>0$,
\be
\ba
\quad&\left|\Prb\left[\left|F(\Psi^{\otimes m},\rho^m)-F_\Psi(\rho)\right|>\lambda\right]-\Prb\left[\left|F(\Psi^{\otimes m},\rho^m)-F_\Psi\left(\sum_{v\in\mathcal V}{p_v\rho_v^{n-4q}}\right)\right|>\lambda\right]\right|\\
&\quad\quad\quad\le q^{\frac{(E+1)^2}2}\exp\left[{-\frac{2q(q+1)}{n}}\right],
\ea
\ee
and thus
\be
\ba
\quad&\Prb\left[\left|F(\Psi^{\otimes m},\rho^m)-F_\Psi(\rho)\right|>\lambda\right]\\
&\le q^{\frac{(E+1)^2}2}\exp\left[{-\frac{2q(q+1)}{n}}\right]+\Prb\left[\left|F(\Psi^{\otimes m},\rho^m)-F_\Psi\left(\sum_{v\in\mathcal V}{p_v\rho_v^{n-4q}}\right)\right|>\lambda\right].
\ea
\label{boundprobaeps}
\ee
With Eq.~(\ref{DeFinettifidereduced}) and Lemma~\ref{lem:fidelitytriangular} we obtain
\be
\left|F\left(\Psi^{\otimes m},\rho^m\right)-F\left(\Psi^{\otimes m},\sum_{v\in\mathcal V}{p_v\rho_v^m}\right)\right|\le q^{\frac{(E+1)^2}2}\exp\left[{-\frac{2q(q+1)}{n}}\right],
\label{boundlemma91}
\ee
where $\Psi^{\otimes m}$ is $m$ copies of the target pure state $\ket\Psi$. 
With the triangular inequality,
\be
\ba
\left|F(\Psi^{\otimes m},\rho^m)-F_\Psi\left(\sum_{v\in\mathcal V}{p_v\rho_v^{n-4q}}\right)\right|&\le\left|F(\Psi^{\otimes m},\rho^m)-F\left(\Psi^{\otimes m},\sum_{v\in\mathcal V}{p_v\rho_v^m}\right)\right|\\
&\quad+\left|F\left(\Psi^{\otimes m},\sum_{v\in\mathcal V}{p_v\rho_v^m}\right)-F_\Psi\left(\sum_{v\in\mathcal V}{p_v\rho_v^{n-4q}}\right)\right|\\
&\le q^{\frac{(E+1)^2}2}\exp\left[{-\frac{2q(q+1)}{n}}\right]\\
&\quad+\left|F\left(\Psi^{\otimes m},\sum_{v\in\mathcal V}{p_v\rho_v^m}\right)-F_\Psi\left(\sum_{v\in\mathcal V}{p_v\rho_v^{n-4q}}\right)\right|,
\ea
\ee
where we used Eq.~(\ref{boundlemma91}) in the last line. With Eq.~(\ref{boundprobaeps}) we obtain, for all $\lambda>0$
\be
\ba
\quad&\Prb\left[\left|F(\Psi^{\otimes m},\rho^m)-F_\Psi(\rho)\right|>\lambda\right]\le q^{\frac{(E+1)^2}2}\exp\left[{-\frac{2q(q+1)}{n}}\right]\\
&\quad\quad+\Prb\left[\left|F\left(\Psi^{\otimes m},\sum_{v\in\mathcal V}{p_v\rho_v^m}\right)-F_\Psi\left(\sum_{v\in\mathcal V}{p_v\rho_v^{n-4q}}\right)\right|>\lambda-q^{\frac{(E+1)^2}2}e^{-\frac{2q(q+1)}{n}}\right].
\ea
\label{almost}
\ee
By linearity of the probabilities, it suffices to bound $\Prb\left[\left|F(\Psi^{\otimes m},\Phi^m)-F_\Psi(\Phi)\right|>\mu\right]$, for $\mu=\lambda-q^{\frac{(E+1)^2}2}\exp\left[{-\frac{2q(q+1)}{n}}\right]$, where $\ket\Phi\in\mathcal S^{n-4q}_{v^{\otimes n-8q}}$, for $\ket v\in\bar{\mathcal H}\otimes\bar{\mathcal H}$, and where $\Phi^m$ is the state obtained from $\ket\Phi\bra\Phi$ by measuring the first $n-4q-m$ subsystems with heterodyne detection and tracing over the purifying subsystems.

\begin{lem}
Let $\ket\Phi\in\mathcal S^{n-4q}_{v^{\otimes n-8q}}$. For all $\epsilon'>0$,
\be
\ba
\quad&\Prb\left[\left|F(\Psi^{\otimes m},\Phi^m)-F_\Psi(\Phi)\right|>\epsilon+\epsilon'\right]\\
&\quad\quad\quad\le2\binom{n-4q}{4q}\exp\left[{-\frac{n-8q}{2m^{4+2E}}\left(\frac{\epsilon^{1+E}\epsilon'}{C_\Psi}-\frac{8qm^{2+E}}{n-4q-m}\right)^2}\right]\\
&\quad\quad\quad\quad\quad\quad+\frac{m(4q+m-1)}{n-4q},
\ea
\ee
where
\be
C_\Psi=\sum_{k,l=0}^E{|\psi_k\psi_l|\left(\frac\epsilon m\right)^{E-\frac{k+l}2}K_\Psi^{1+\frac{k+l}2}\sqrt{2^{|l-k|}\binom{\max{(k,l)}}{\min{(k,l)}}}}\underset{\epsilon\rightarrow0}{\longrightarrow}|\psi_E|^2K_\Psi^{1+E}.
\ee
\label{interbound}
\end{lem}

\begin{leftbar}
\begin{proof} 
Let $\alpha_1,\dots,\alpha_{n-4q-m}$ be samples obtained by measuring the first $n-4q-m$ subsystems of $\ket\Phi\bra\Phi$ with heterodyne detection. We have~(\ref{tildeF2app})
\be
F_\Psi(\Phi)=\left[\frac{1}{n-4q-m}\smashoperator{\sum_{i=1}^{n-4q-m}}{f_\Psi\left(\alpha_i,\frac\epsilon{mK_\Psi}\right)}\right]^m,
\ee
and
\be
\ba
\,&\left|F(\Psi^{\otimes m},\Phi^m)-F_\Psi(\Phi)\right| \le\left|F(\Psi^{\otimes m},{\Phi}^m)-F(\Psi^{\otimes m},\ket{v}\bra{v}^{\otimes m})\right|\\
  &\quad\quad\quad+\left|F(\Psi^{\otimes m},\ket{v}\bra{v}^{\otimes m})
   -\left(\underset{\beta\leftarrow Q_{\ket v\bra v}}{\mathbb E}%
                                        \left[f_\Psi\left(\beta,\frac\epsilon{mK_\Psi}\right)\right]\right)^m\right|\\
  &\quad\quad\quad+\left|\left(\underset{\beta\leftarrow Q_{\ket v\bra v}}{\mathbb E}%
                                        \left[f_\Psi\left(\beta,\frac\epsilon{mK_\Psi}\right)\right]\right)^m
        -F_\Psi(\Phi)\right|\\
 &\quad\quad\quad\quad\quad\quad\quad\quad\quad\quad\quad=\left|F(\Psi^{\otimes m},{\Phi}^m)-F(\Psi^{\otimes m},\ket{v}\bra{v}^{\otimes m})\right|\\
  &\quad\quad\quad+\left|F(\Psi,\ket{v}\bra{v})^m
   -\left(\underset{\beta\leftarrow Q_{\ket v\bra v}}{\mathbb E}%
                                        \left[f_\Psi\left(\beta,\frac\epsilon{mK_\Psi}\right)\right]\right)^m\right|\\
  &\quad\quad\quad+\left|\left(\underset{\beta\leftarrow Q_{\ket v\bra v}}{\mathbb E}%
                                        \left[f_\Psi\left(\beta,\frac\epsilon{mK_\Psi}\right)\right]\right)^m
        -\left(\frac{1}{n-4q-m}\sum_{i=1}^{n-4q-m}{f_\Psi\left(\alpha_i,\frac\epsilon{mK_\Psi}\right)}\right)^m\right|\\
 &\quad\quad\quad\quad\quad\quad\quad\quad\quad\quad\le\left|F(\Psi^{\otimes m},{\Phi}^m)-F(\Psi^{\otimes m},\ket{v}\bra{v}^{\otimes m})\right|\\
  &\quad\quad\quad+m\left|F(\Psi,\ket v\bra v)
    -\underset{{\beta}\leftarrow Q_{\ket v{\bra v}}}{\mathbb E}\left[f_\Psi\left(\beta,\frac\epsilon{mK_\Psi}\right)\right]\right|\\
&\quad\quad\quad+m\left|\underset{{\beta}\leftarrow Q_{\ket v{\bra v}}}{\mathbb E}\left[f_\Psi\left(\beta,\frac\epsilon{mK_\Psi}\right)\right]-\frac{1}{n-4q-m}\sum_{i=1}^{n-4q-m}{f_\Psi\left(\alpha_i,\frac\epsilon{mK_\Psi}\right)}\right|,
\ea
\label{threebounds}
\ee
where we used Lemma~\ref{lem:simplem}. We bound these last three terms in the following.\\

\noindent When selecting at random $m$ subsystems from an almost-i.i.d.\@ state over $n-4q$ subsystems which is i.i.d.\@ on $n-8q$ subsystems, the probability that all of the selected states are from the $n-8q$ i.i.d.\@ subsystems is
\be
\frac{\binom{n-8q}{m}}{\binom{n-4q}{m}}=\frac{(n-8q)(n-8q-1)\dots(n-8q-m+1)}{(n-4q)(n-4q-1)\dots(n-4q-m+1)},
\ee
and we have
\be
\ba
1-\frac{(n-8q)(n-8q-1)\dots(n-8q-m+1)}{(n-4q)(n-4q-1)\dots(n-4q-m+1)}&\le1-\frac{(n-8q-m+1)^m}{(n-4q)^m}\\
&=1-\left(1-\frac{4q+m-1}{n-4q}\right)^m\\
&\le\min{\left(1,\frac{m(4q+m-1)}{n-4q}\right)}\\
&\le\frac{m(4q+m-1)}{n-4q},
\ea
\label{boundmNQ}
\ee
where we used $1-(1-x)^a\le ax$ for all $a\ge1$ and $x\in[0,1]$. In particular, for $\ket\Phi\in\mathcal S^{n-4q}_{v^{\otimes n-8q}}$, and $\Phi^m$ its reduced state over $m$ modes chosen at random, we have
\be
\Phi^m=\ket v\bra v^{\otimes m},
\ee
with probability greater than $1-\frac{m(4q+m-1)}{n-4q}$, where we used the definition of $\mathcal S^{n-4q}_{v^{\otimes n-8q}}$, and Eq.~(\ref{boundmNQ}). Using Lemma~\ref{lem:fidelitytriangular}, the \textit{first term} in Eq.~(\ref{threebounds}) vanishes with probability greater than:
\be
1-\frac{m(4q+m-1)}{n-4q}.
\label{finalbound1}
\ee
The bound for the \textit{second term} is given by Corollary~\ref{corofidelity} applied to the state $\ket v$, for $\eta=\frac\epsilon{mK_\Psi}$:
\be
m\left|F(\Psi,\ket v\bra v)
    -\underset{{\beta}\leftarrow Q_{\ket v{\bra v}}}{\mathbb E}\left[f_\Psi\left(\beta,\frac\epsilon{mK_\Psi}\right)\right]\right|\le\epsilon.
\label{finalbound2}
\ee
The bound for the \textit{third term} is probabilistic, given by Eq.~(\ref{Hoeffdingpsiv}), for $\eta=\frac\epsilon{mK_\Psi}$ and $\mu=\frac{\epsilon'}m$. For all $\epsilon'>0$,
\be
\ba
\quad&\underset{\bm{\alpha}}{\Prb}\left[
  \left|\frac{1}{n-4q-m}\sum_{i=1}^{n-4q-m}{f_\Psi\left(\alpha_i,\frac\epsilon{mK_\Psi}\right)}
  \mathrlap{-}\underset{\quad\beta\leftarrow Q_{\ket v\bra v}}{\mathbb E}\left[f_\Psi\left(\beta,\frac\epsilon{mK_\Psi}\right)\right]
  \right|
 \ge\frac{\epsilon'}m\right] \\
&\quad\quad\quad\leq2\binom{n-4q}{4q}\exp\left[{-\frac{n-8q}2\left(\frac{\epsilon^{1+E}\epsilon'}{m^{2+E}K_\Psi^{1+E}M_\Psi(\frac\epsilon{mK_\Psi})}-\frac{8q}{n-4q-m}\right)^2}\right].
\ea
\label{finalbound3}
\ee

We now bring together the previous bounds in order to prove Lemma~\ref{interbound}. Combining Eqs.~(\ref{threebounds}), (\ref{finalbound1}), (\ref{finalbound2}) and (\ref{finalbound3}) yields
\be
\ba
\,&\Prb\left[\left|F(\Psi^{\otimes m},\Phi^m)-F_\Psi(\Phi)\right|>\epsilon+\epsilon'\right]\\
&\quad\le\underset{\bm{\alpha}}{\Prb}\left[
  \left|\frac{1}{n-4q-m}\sum_{i=1}^{n-4q-m}{f_\Psi\left(\alpha_i,\frac\epsilon{mK_\Psi}\right)}
  \mathrlap{-}\underset{\quad\beta\leftarrow Q_{\ket v\bra v}}{\mathbb E}\left[f_\Psi\left(\beta,\frac\epsilon{mK_\Psi}\right)\right]
  \right|
 \ge\frac{\epsilon'}m\right] \\
 &\quad\quad\quad+\frac{m(4q+m-1)}{n-4q}\\
&\quad\le2\binom{n-4q}{4q}\exp\left[{-\frac{n-8q}2\left(\frac{\epsilon^{1+E}\epsilon'}{m^{2+E}K_\Psi^{1+E}M_\Psi(\frac\epsilon{mK_\Psi})}-\frac{8q}{n-4q-m}\right)^2}\right]\\
 &\quad\quad\quad+\frac{m(4q+m-1)}{n-4q}\\
&\quad=2\binom{n-4q}{4q}\exp\left[{-\frac{n-8q}{2m^{4+2E}}\left(\frac{\epsilon^{1+E}\epsilon'}{C_\Psi}-\frac{8qm^{2+E}}{n-4q-m}\right)^2}\right]+\frac{m(4q+m-1)}{n-4q},
\ea
\ee
where
\be
\ba
C_\Psi&=K_\Psi^{1+E}M_\Psi\left(\frac\epsilon{mK_\Psi}\right)\\
&=\sum_{k,l=0}^E{|\psi_k\psi_l|\left(\frac\epsilon m\right)^{E-\frac{k+l}2}K_\Psi^{1+\frac{k+l}2}\sqrt{2^{|l-k|}\binom{\max{(k,l)}}{\min{(k,l)}}}}\underset{\epsilon\rightarrow0}{\longrightarrow}|\psi_E|^2K_\Psi^{1+E}.
\ea
\ee
\end{proof}
\end{leftbar}

\noindent Combining Eq.~(\ref{almost}) and Lemma~\ref{interbound}, we finally obtain
\be
\ba
\Prb &\left[\left|F(\Psi^{\otimes m},\rho^m)-F_\Psi(\rho)\right|>\epsilon+\epsilon'+q^{\frac{(E+1)^2}2}e^{-\frac{2q(q+1)}{n}}\right]\\
&\le q^{\frac{(E+1)^2}2}\exp\left[{-\frac{2q(q+1)}{n}}\right]+2\binom{n-4q}{4q}\exp\left[{-\frac{n-8q}{2m^{4+2E}}\left(\frac{\epsilon^{1+E}\epsilon'}{C_\Psi}-\frac{8qm^{2+E}}{n-4q-m}\right)^2}\right]\\
&\quad\quad\quad\quad\quad+\frac{m(4q+m-1)}{n-4q}.
\ea
\label{appprobafinale}
\ee
Setting $P_{\text{Hoeffding}}=2\binom{n-4q}{4q}\exp\left[{-\frac{n-8q}{2m^{4+2E}}\left(\frac{\epsilon^{1+E}\epsilon'}{C_\Psi}-\frac{8qm^{2+E}}{n-4q-m}\right)^2}\right]$, $P_{\text{choice}}=\frac{m(4q+m-1)}{n-4q}$ and $P_{\text{deFinetti}}=q^{\frac{(E+1)^2}2}\exp\left[{-\frac{2q(q+1)}{n}}\right]$, we obtain
\be
\Prb\left[\left|F(\Psi^{\otimes m},\rho^m)-F_\Psi(\rho)\right|>\epsilon+\epsilon'+P_{\text{deFinetti}}\right]\le P_{\text{deFinetti}}+P_{\text{choice}}+P_{\text{Hoeffding}}.
\ee

\medskip

Until now we have assumed $\rho^n\in\mathcal S^n_{\bar{\mathcal H}^{\otimes n-q}}$. By section~\ref{app:Etestsym},
\be
\Pr\left[\mathcal F_q^n\cap\mathcal T_{\le s}^k\right]\le P_{\text{support}}.
\ee
where $\mathcal F_q^n$ is the event that the projection of $\rho^n$ (the remaining state after the support estimation step) onto $\mathcal S^n_{\bar{\mathcal H}^{\otimes n-q}}$ fails, where $\mathcal T_{\le s}^k$ is the event that at most $s$ of the $k$ values $\beta_i$ from the support estimation step satisfy $|\beta_i|^2>E$, and where $P_{\text{support}}=8k^{3/2}\exp\left[{-\frac{k}9\left(\frac qn-\frac{2s}k\right)^2}\right]$. With the union bound we thus obtain
\be
\Prb\left[\left(\left|F(\Psi^{\otimes m},\rho^m)-F_\Psi(\rho)\right|>\epsilon+\epsilon'+P_{\text{deFinetti}}\right)\cap\mathcal T_{\le s}^k\right]\le P_{\text{support}}+P_{\text{deFinetti}}+P_{\text{choice}}+P_{\text{Hoeffding}},
\ee
where
\be
\ba
P_{\text{support}}&=8k^{3/2}\exp\left[{-\frac{k}9\left(\frac qn-\frac{2s}k\right)^2}\right]\\
P_{\text{deFinetti}}&=q^{\frac{(E+1)^2}2}\exp\left[{-\frac{2q(q+1)}{n}}\right]\\
P_{\text{choice}}&=\frac{m(4q+m-1)}{n-4q}\\
P_{\text{Hoeffding}}&=2\binom{n-4q}{4q}\exp\left[{-\frac{n-8q}{2m^{4+2E}}\left(\frac{\epsilon^{1+E}\epsilon'}{C_\Psi}-\frac{8qm^{2+E}}{n-4q-m}\right)^2}\right].
\ea
\label{eq:scalings}
\ee

\qed\\

\noindent The variables $\epsilon,\epsilon',n,m,q,k,s,E$ are free parameters of the protocol. Let us fix, e.g., $E=O(1)$, $s=O(1)$, \mbox{$n=O\left(m^{19+8E}\right)$}, $k=O(m^{19+8E}), q=O\left(m^{10+4E}\right)$, and $\epsilon=\epsilon'=O(\frac1m)$.
Then, with Theorem~\ref{thVUCVQC}, either the estimate $F_\Psi(\rho)$ of the fidelity $F(\Psi^{\otimes m},\rho^m)$ is polynomially precise (in $m$), or the score at the support estimation step is higher than $s$, with polynomial probability (in $m$), by plugging the different scalings in Eq.~(\ref{eq:scalings}).


\bibliography{bibliography}


\end{document}